\documentclass[11pt,reqno]{article}
\usepackage[round]{natbib}
\usepackage{amsmath,amsfonts,amsthm}
\usepackage[margin=1in]{geometry}
\usepackage[normalem]{ulem}

\usepackage[hyphens]{url}
\usepackage{hyperref}
\usepackage{graphicx}
\usepackage{verbatim}
\usepackage{enumitem} 
\usepackage{color}
\usepackage{tikz}
\usetikzlibrary{decorations.fractals}
\usetikzlibrary{patterns,arrows,shapes,decorations.pathreplacing}

\numberwithin{equation}{section}

\newtheorem{theorem}{Theorem}[section]
\newtheorem{assumption}{Assumption}[section]
\newtheorem{corollary}{Corollary}[section]
\newtheorem{lemma}{Lemma}[section]
\newtheorem{proposition}{Proposition}[section]

\newtheorem{definition}{Definition}[section]
\newtheorem{remark}{Remark}[section]

\newcommand{\propref}{Proposition~\ref}
\newcommand{\thmref}{Theorem~\ref}

\renewcommand{\P}{\mathbb{P}}

\newcommand{\R}{\mathbb{R}}
\newcommand{\E}{\mathbb{E}}

\newcommand{\N}{\mathbb{N}}
\newcommand{\F}{\mathcal{F}}

\newcommand{\cC}{\mathcal{C}}

\newcommand{\A}{\mathcal{A}}

\newcommand{\fL}{\mathfrak{L}}
\newcommand{\eps}{\varepsilon}

\newcommand{\nada}[1]{}

\title{Consumption, Investment, and Healthcare with Aging\thanks{We thank for helpful comments seminar participants at Collegio Carlo Alberto, ETH Z\"urich, University of Limerick, Alfred Renyi Institute, National Central University in Taiwan, the QMF conference at UTS Sydney, the Congress of the Bachelier Finance Society, and the University of Colorado at Boulder.}
}

\author{Paolo Guasoni\thanks{Boston University, Department of Mathematics and Statistics, 111 Cummington Mall, Boston, MA 02215, USA, and Dublin City University, School of Mathematical Sciences, Glasnevin, Dublin 9, Ireland, email: \texttt{paolo.guasoni@dcu.ie}. Partially supported by the ERC (278295), NSF (DMS-1412529), and SFI (16/SPP/3347 and 16/IA/4443).} \and Yu-Jui Huang \thanks{Department of Applied Mathematics, University of Colorado, Boulder, CO 80309, USA, email: \texttt{yujui.huang@colorado.edu}. Partially supported by NSF (DMS-1715439) and the University of Colorado (11003573).}
}

\date{}

\begin{document}
\maketitle

\nada{
\begin{center}
\Large
Preliminary and incomplete.\\
Please do not quote or redistribute without permission.
\end{center}
\bigskip
}

\begin{abstract}
This paper solves the problem of optimal dynamic consumption, investment, and healthcare spending with isoelastic utility, when natural mortality grows exponentially to reflect Gompertz' law and investment opportunities are constant. Healthcare slows the natural growth of mortality, indirectly increasing utility from consumption through longer lifetimes. 
Optimal consumption and healthcare imply an endogenous mortality law that is asymptotically exponential in the old-age limit, with lower growth rate than natural mortality. Healthcare spending steadily increases with age, both in absolute terms and relative to total spending. 
The optimal stochastic control problem reduces to a nonlinear ordinary differential equation with a unique solution, which has an explicit expression in the old-age limit. The main results are obtained through a novel version of Perron's method.

\end{abstract}

\smallskip

\smallskip


\newpage

\section{Introduction}
The steady rise in both incomes and life expectancy over the past century 
alludes to tantalizing links between healthcare, wealth, and mortality, which are the subject of heated debate. Understanding the relative roles of healthcare, wealth, and medical progress in explaining longevity gains is as important as it is difficult (see \cite{CutlerDeatonLleras-Muney2006} for a survey). Yet, with few exceptions, models of optimal consumption and investment have largely shied away from mortality and healthcare, leaving a wide gap between idealized theoretical settings and realistic empirical studies.

Healthcare is different from consumption in typical goods and services: it often causes immediate pain, and is justified only by its expected effect in reversing or delaying the onset of disease and, ultimately, death. Healthcare does not directly generate utility, but, by reducing mortality risk, it extends the lifetime over which consumption yields utility.

Mortality (the probability that someone alive today dies next year) displays an approximate exponential growth with age, an observation that has remained remarkably stable since its discovery by \cite{Gompertz1825}, even as mortality has steadily declined at all age groups (Figure \ref{fig:mortality}). A central question is to which extent such decline can be ascribed to the availability and the optimal use of healthcare, and a satisfactory answer hinges on the predictions of a model in which healthcare choices and their resulting mortality are endogenous.
This is the goal of this paper.

Our model focuses on a representative household that makes consumption, investment, and healthcare spending decisions to maximize welfare. 
Consumption generates utility, healthcare reduces the growth of mortality below its natural, constant rate of Gompertz' law, while investment helps increasing wealth. Optimality is reached when their marginal values are the same. 

A careful representation of the impact of death is a critical issue in models with endogenous mortality: the ostensibly natural approach of expected lifetime utility leads to  potential preference for death over life and violates the invariance of preferences to affine transformations. To avoid this pitfall, we assume that a death leaves the surviving household with a fraction of its previous wealth, but the same mortality rate. This representation is consistent with the interpretation of wealth as present value of current assets and future cash flows, and of the surviving household as a spouse in roughly the same age group. It also preserves affine invariance and unconditional preference of life over death.

Equally important is the impact of healthcare on mortality growth -- the \emph{efficacy} function. An unwise specification, in which sufficiently high spending can arrest and reverse aging, can lead to the implausible (and counterfactual) conclusion that early healthcare interventions can bring mortality to zero, leaving a household immortal albeit less wealthy. We exclude such dubious outcomes through two assumptions: first, the efficacy depends on healthcare spending relative to wealth, thereby emphasizing the opportunity costs (such as foregone income) of healthcare and healthy behaviors. Second, we posit that the same amount of healthcare is more effective when health is worse, i.e., when mortality is higher.

This paper contributes to the literature both mathematically and economically. On the technical side, the model leads to an optimal control problem which features, in addition to the usual consumption and investment, an arrival rate of jumps (deaths) that is partially controlled (by healthcare) through a state variable (mortality). As such a control problem with jumps, it normally corresponds to a Hamilton-Jacobi-Bellman (HJB) integro-differential equation. Yet, exploiting a scaling property of the value function, we are able to reduce the HJB equation to a nonlinear ordinary differential equation. Although such an equation does not admit explicit solutions, we prove that it has a single global solution through an unconventional version of Perron's method.  

Perron's method generally refers to the construction of a viscosity solution by interposition between a supersolution and a subsolution, which are typically more tractable (cf. \cite{JS12, BZ15}). By contrast, we apply Perron's method not to a single pair of super- and subsolutions, but to a collection thereof (Definition~\ref{def:Pi}). In addition to identifying the value function as the unique solution to the reduced nonlinear HJB equation (Theorem~\ref{thm:main 3}), this approach also delivers powerful estimates on the value function (Theorem~\ref{thm:main 4}), which in turn yield that Gompertz' law holds asymptotically in the old-age limit. Such additional insights entail some challenges, such as verifying the subsolution property (Propositions \ref{prop:subsolution} and \ref{prop:strict concave}): as the equation involved is of the first order, the regularizing property of a second-order term is not available to establish strict concavity, in contrast to standard viscosity applications.

\begin{figure}[t]
\centering
\includegraphics[width=.9\textwidth]{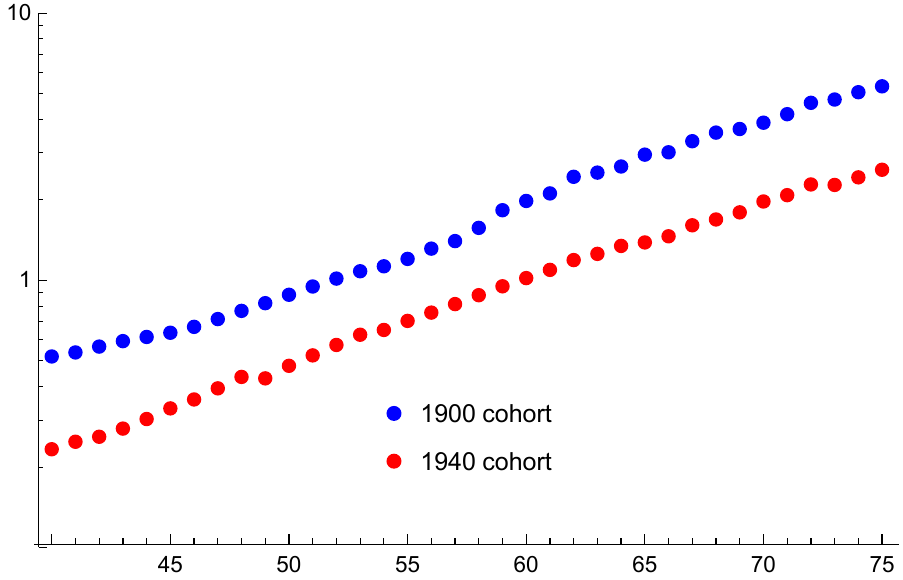}
\caption{\label{fig:mortality}
Mortality rates (vertical axis, in logarithmic scale) at adults ages for the birth cohorts of years 1900 and 1940 (from top to bottom) in the United States (both males and females). Source: Berkeley Human Mortality Database.
}
\end{figure}

Proving that the solution to the differential equation is indeed the value function (verification) also requires new tools. Standard results do not apply, as they do not support controllable jump rates, and, in extant results that support them \cite[chapters 20 and 21]{cohen2015stochastic}, jump rates depend on a control variable in a compact set, while in the present problem they depend on a \emph{state} variable, and the control has unbounded support. The technical challenges of such an extension are overcome in Theorems \ref{thm:verification} and \ref{thm:verification with stock}, which establish a verification result that exploits the properties of the law of the candidate optimal process.\footnote{A related potential approach to dealing with controlled jumps is stochastic Perron's method (a
verification type result without smoothness required), as recently employed in \cite{BL17}.}

On the economic side, the effect of (exogenous) mortality on household decisions has long been recognized. \cite{Yaari1965} solves the problem of consumption and investment in a complete market and finds that annuitization is optimal in the absence of a bequest motive. Vice versa, \cite{Richard1975} proves that buying life insurance is optimal as a hedge against loss of future earnings. In the optimal consumption framework, \cite{rosen1988value} investigates the welfare changes from gains in life expectancy, observing the potential preference for death implied by some utility functions, while \cite{ShepardZeckhauser1984} estimate the reservation price of reduced mortality risk, which hints at the demand for healthcare. Overall, this literature investigates the marginal value of extending life, yet without incorporating this possibility endogenously.

Endogenous mortality dates back to \cite{Grossman1972}, who introduces the concept of health capital, assuming that its depreciation rate increases with age, and that death occurs when health drops below a given threshold. The limits of this approach are that death is a deterministic, perfectly foreseen event, and that health has constant returns to scale, so with bounded health depreciation a wealthy individual could choose to live forever. \cite{EhrlichChuma1990} overcome the latter issue assuming decreasing returns to scale, while \cite{Ehrlich2000} makes death uncertain by replacing health capital with ``self-protection'' expenditures that reduce mortality.
\cite{Yogo2016} and \cite{HugonnierPelgrinSt-Amour2012} combine these approaches, modeling mortality as a function of health capital that depreciates over time, and specifying the accumulation of health capital as a concave function of health-expenditures \emph{relative} to health capital. By contrast, \cite{hall2007value} emphasize the dual role of health in determining mortality and additively increasing utility from consumption.

As this literature has evolved toward richer and more realistic models, Gompertz' law of exponential mortality growth, however, has remained a conspicuous absence. This paper makes this feature a central element: taking the mortality rate as sole state variable in addition to wealth, determined jointly by Gompertz' law and healthcare choices, we investigate the mutual response between healthcare spending and mortality growth, and the lower mortality rates that result.

Several important implications are brought by our analysis.
First, it identifies the marginal efficacy of optimal health spending as inversely proportional to the elasticity of consumption with respect to mortality. Because in the model such elasticity increases with mortality, and healthcare spending has diminishing returns (decreasing marginal efficacy), it follows that health spending relative to wealth increases with age and mortality, converging to a maximum finite rate in the old-age limit.
Second, health spending is nearly negligible in youth, but it rises rapidly with age, outpacing the growth in consumption and taking a larger share of total spending. At very old ages the trend reverses, as health spending rate stabilizes while consumption continues to rise with mortality. The latter effect, however, becomes visible only at ages that are not reached by most individuals.
Third, the model generates an endogenous mortality curve in which its natural growth is reduced by healthcare.  Importantly, endogenous mortality is also close to exponential, and asymptotically exponential in the old-age limit, thereby confirming the empirical observation that Gompertz' law has survived two centuries of medical progress.
The reduction in mortality growth depends on the efficacy of healthcare as well as the elasticity of intertemporal substitution (EIS), but not on other quantities.

The rest of the paper is organized as follows: Section 2 first discusses in detail the assumptions on preferences, mortality, and healthcare, then provides rigorous definitions of the model and its probabilistic structure. Section 3 presents the main results in order of complexity, first in a baseline model with neither aging nor healthcare, then adding aging, and finally in the complete setting with both aging and healthcare. {Section 4 incorporates risky assets into the main model, while Section 5 calibrates the main model and discusses the implications. Section 6 concludes.}

\section{The Model}

The main model aims to understand optimal consumption and healthcare spending in relation to mortality, 
with a focus on a household seeking to maximize total welfare. {Section~\ref{sec:risky} extends the analysis to include investment in risky assets.}

\subsection{Preferences and Lifetimes}
A natural starting point are the familiar time-additive preferences, in which expected welfare has the representation
\begin{equation}\label{eq:simobj}
\E\left[\int_0^\tau e^{-\delta t}U(X_t c_t)dt \right]
,
\end{equation}
where $\tau$ denotes the lifetime,  $U:\mathbb R_+ \to \mathbb R$ is a utility function (i.e., increasing and concave), the parameter $\delta\ge 0$ captures time-preference, and $c_t$ represents the rate of consumption per unit of time, as a fraction of current wealth $X_t$, interpreted as the household's net worth, which includes the present value of future income, not specified separately.

The limit of this approach is that it tacitly reduces death to a change in preferences, assuming that utility equals zero in the afterlife $[\tau,\infty)$, and implying that death is preferable to negative utility, as recognized by \cite{ShepardZeckhauser1984}, \cite{rosen1988value}, \cite{BommierRochet2006}, \cite{hall2007value}, \cite{Bommier2010}.

Yet, a more appealing approach than costless afterlife utility is to note that households have concrete bequest motives which center on the welfare of similar individuals. For example, upon his death a man may leave behind a wife in a similar age group, hence with a similar mortality rate. Such a household makes consumption and healthcare spending choices that account for the welfare of both spouses over their lifetimes. Larger households face even more complex choices, which involve the lives of several people.

Striking a balance between realism and tractability, suppose that a household experiences a sequence of deaths at times $(\tau_n)_{n\in\N_0}$ with $0:=\tau_0<\tau_1<\dots<\tau_n\uparrow \infty$ a.s. and that after each death the surviving household members inherit a fraction $\zeta\in [0,1]$ of wealth, while retaining the same mortality. This assumption is clearly a simplification, as in reality households include only a few members, but this flaw is mitigated by the time-preference parameter $\delta$, whereby the first few lifetimes account for most of the expected utility.

With this assumption, denoting by $X_t$ the household wealth at time $t$ if no deaths have occurred, the actual household wealth after the $n$-th death is $\zeta^n X_t$, and total welfare becomes
\begin{equation}
\E\left[ \sum_{n=0}^\infty \int_{\tau_{n}}^{\tau_{n+1}} e^{-\delta t} U(\zeta^n X_t c_t) dt\right]
.
\end{equation}
Furthermore, the discussion henceforth focuses on the isoelastic class
\begin{equation}\label{eq:power_utility}
U(x) = \frac{x^{1-\gamma}}{1-\gamma}
\qquad 0<\gamma\ne 1
\end{equation}
which, in the absence of healthcare and mortality, generates consumption policies proportional to wealth. Such a property is attractive because empirical consumption-wealth ratios do not display any significant secular trend. 



In contrast to the lifetime-horizon approach described earlier, this model does not equate death to a change in preferences, but rather to a loss for the surviving household, while leaving preferences unchanged. The parameter $\zeta$ controls the severity of the loss: $\zeta=1$ implies immortality (common in the literature as uncommon in reality), as the arrivals of $\tau_n$ are inconsequential; at the other extreme, with $\zeta=0$ death implies a total loss, after which only zero spending is possible, leading to a constant utility rate of $U(0)$.

In general, $\zeta\in[0,1]$ crudely summarizes the combined economic effects of death, which include inheritance and estate taxes, loss of pensions and annuities, foregone future income, and a myriad of other actual or opportunity costs. 
(A loss of future cash flow is equivalent to a loss in wealth, assuming that the cash flow is replicable.) 
Although the model does not include explicitly life-insurance and annuity contracts, the parameter $\zeta$ can also be thought of as a measure of protection of the household wealth against mortality losses, with full protection for $\zeta=1$ and no protection for $\zeta=0$.

The model does not distinguish between the relative impact of household components with different ages, as the focus of this paper is not on household structure but rather on the tradeoff between consumption and healthcare, which is now introduced.

\subsection{Healthcare and Mortality}

The household is homogeneous, in that all members share the same mortality $M_t$ starting at the initial level $M_0 := m_0$. In the absence of healthcare, mortality grows exponentially, consistently with the classical \cite{Gompertz1825} law
\begin{equation}
dM_t = \beta M_t dt 
.
\end{equation}
Healthcare spending reduces mortality growth according to an \emph{efficacy} function $g: \mathbb R_+ \to \mathbb R_+$ of $h_t$, the spending rate in healthcare as a fraction of household wealth:
\begin{equation}\label{eq:mortdyn}
dM_t = (\beta - g(h_t))M_t dt
.
\end{equation}

The efficacy function $g$ is assumed {strictly increasing and} concave, which reflects the diminishing returns from increased health expenditure. In addition, $g(0)=0$, which identifies $\beta$ as the \emph{natural} rate at which mortality grows in the absence of healthcare expenditures. Finally $g$ is defined only on the positive real line, consistent with the interpretation of health investment as irreversible.\footnote{See for example \cite{Grossman1972}, \cite{EhrlichChuma1990}, \cite{hall2007value}.}

The assumption that healthcare expenses affect mortality growth relative to wealth rather than in absolute terms emphasizes the lost income and earning opportunities resulting from healthcare usage.\footnote{For example, \citeauthor{Smith1999} (\citeyear{Smith1999}, \citeyear{Smith2005}) reports ill health as a leading cause of early retirement.}
For example, for households whose wealth is dominated by the value of future income, the lost income from healthcare usage is approximately proportional to wealth. In addition, means-tested subsidies and income taxes on health-insurance premiums effectively make the same medical procedures cheaper for poorer households, and the assumption of proportionality approximates this dependence with a linear relation. 
Likewise, \cite{chetty2016association} recently find that life expectancy is significantly correlated with health behaviors but not with access to medical care.\footnote{In their words, \emph{geographical differences in life expectancy for individuals in the lowest income quartile were significantly correlated with health behaviors such as smoking [...], but were not significantly correlated with access to medical care, physical environmental factors, income inequality, or labor market conditions.}}
In reality the determinants of healthcare spending on mortality are complex \citep{CutlerDeatonLleras-Muney2006}, and the relative importance of proportional and absolute components is largely an empirical question. The present simplification offers a plausible and parsimonious approximation that focuses on proportional costs.

A related important reason to consider proportional costs is to avoid the unrealistic implication that wealth buys immortality. Indeed, if $h_t$ in \eqref{eq:mortdyn} were to represent an absolute amount of health expenditures, wealthy individuals could effectively reduce mortality to zero through early health expenses while maintaining non-zero consumption. In reality, life expectancy in the top 1\% income percentile is about five years higher than for median incomes \citep[Figure 2]{chetty2016association}.

\subsection{Savings}

In the basic version of the model, the household leaves savings in a safe asset which earns a constant rate $r$, with no other financial or insurance contracts available. In particular, the household does not have access to life-insurance contracts that pay out in the event of death. This assumption is consistent with the interpretation of wealth as inclusive of future income, which in practice can be hedged only in rather limited amounts.

With these assumptions, at each time $t$ the household spends at rates $c_t$ in consumption and $h_t$ in healthcare, while earning a constant interest rate $r$ on wealth.
If $N_t$ denotes the number of deaths up to time $t$, regulated by the mortality dynamics in \eqref{eq:mortdyn}, household wealth $\Xi_t = X_t \zeta^{N_t}$ incorporating death losses evolves as:
\begin{equation}
\frac{d\Xi_t}{\Xi_t} = (r - c_t - h_t)dt  - (1-\zeta)dN_t
.
\end{equation}
Note that the only source of randomness is the arrival of deaths, without which the model reverts to a deterministic consumption-investment problem, in which the optimal policy is to consume at a rate proportional to wealth.

After the description and motivation of the main model provided here, the next section proceeds with the mathematical details required for the precise statement of the main result.

\subsection{Definitions and Notation}\label{subsec:setup}

The rigorous formulation of the model starts with the probability space $(\Omega,\F,\P)$, which supports a sequence $\{Z_n\}_{n\in\N}$ of independent, identically distributed random variables with an exponential law $\P(Z_n> z) = e^{-z}$ for all $z\ge 0$ and $n\in\N$. 
(These random times are interpreted as mortality-adjusted times of death, as defined below.)

Denote by $L^{1,+}_\text{loc}$ the collection of all nonnegative locally integrable functions $f:\R_+\to\R_+$, and note that $L^{1,+}_\text{loc}$ is metrizable, thus a Borel space. Define also
\[
\mathfrak L := \left\{\{f_n\}_{n\in \N_0}\ \middle|\  f_0\in  L^{1,+}_\text{loc},\ f_n = f(Z_1,...,Z_n)\ \hbox{for some Borel}\ f:\R^n_+\to L^{1,+}_\text{loc}\right\},
\]
which represents the family of sequences of consumption-healthcare policies, with the policy after the $n$-th death depending possibly on the previous $n$ times, in addition to calendar time.

Consider a nondecreasing concave function $g:\R_+\to\R_+$ with $g(0)=0$. For any $(t,m)\in \R^2_+$ and $h\in L^{1,+}_\text{loc}$, let $M^{t,m,h}$ be the deterministic process defined by 
\begin{equation}\label{M with h}
dM^{t,m,h}_s = M^{t,m,h}_s \left[\beta  - g(h(s))\right]ds,\quad M^{t,m,h}_t=m,
\end{equation}
where $\beta\ge 0$ is a fixed constant. 
Next, the arrival times of deaths are defined in terms of the random variables $\{Z_n\}_{n\in\N}$. For any $m\ge 0$ and $\{h_n\}\in \mathfrak L$, construct recursively a sequence $\{\tau_n\}_{n\in \N_0}$ of random times as follows: first, set $\tau_{0} := 0$ and $m_0:=m$; then, for each $n\ge 0$, define
\begin{equation}\label{tau's}
\tau_{n+1} := \inf\left\{t\ge \tau_{n}\ \middle|\ \int_{\tau_{n}}^t M^{\tau_{{n}},m_{n},h_{n}}_s ds \ge Z_{n+1}\right\},\quad m_{n+1}:= M^{\tau_{n},m_{n},h_{n}}_{\tau_{{n+1}}}.
\end{equation}
Now introduce the counting process $\{N_t\}_{t\ge 0}$:
\begin{equation}\label{counting process}
N_t := n\quad \hbox{for}\ t\in[\tau_{{n}},\tau_{{n+1}}),
\end{equation}
and observe from the construction of $\{\tau_{n}\}_{n\in \N_0}$ that
\begin{equation}\label{tau>t}
\P\left(N_t =n\ \middle|\ Z_1,...,Z_n \right) = \P\left(t\in [\tau_{n},\tau_{{n+1}})\ \middle|\ Z_1,...,Z_n\right) = \exp\left(-\int_{\tau_{n}}^{t} M^{\tau_{n},m_n,h_n}_s ds \right)1_{\{t\ge \tau_{n}\}}
,
\end{equation}
which means that the mortality rate of $\tau_{n+1}$ at time $t$ is precisely $M^{\tau_{n},m_n,h_n}_t$, as required.

Consider now the collection of processes:
\begin{equation}\label{A}
\A := \left\{c_t = \sum_{n=0}^\infty c_n(t) 1_{\{\tau_{n}\le t<\tau_{n+1}\}},\ h_t = \sum_{n=0}^\infty h_n(t) 1_{\{\tau_{n}\le t<\tau_{n+1}\}}\ \middle|\ \{c_n\}, \{h_n\}\in\mathfrak L\right\}.
\end{equation}
The construction of $\A$ is understood as follows: first, use $\{h_n\}\in\mathfrak L$ to construct $\{\tau_n\}$ as in \eqref{tau's}; then, use $\{c_n\},\{h_n\}\in\mathfrak L$ to define the processes $c_t$ and $h_t$.  Then, the process $M^{t,m,h}$ is defined as in \eqref{M with h} for any $h$ as in \eqref{A}.

\subsection{Problem Formulation} 

With initial wealth $x\ge 0$ and initial mortality rate $m\ge 0$ at time $t\ge 0$, an household at each time $s\ge t$ chooses the rates of spending in consumption ($c_s\ge 0$) and healthcare ($h_s\ge 0$).  With savings earning the safe rate $r$, household wealth $X^{t,x,c,h}_s$ before mortality losses evolves as
\begin{equation}\label{wealth}
dX^{t,x,c,h}_s = X^{t,x,c,h}_s[r - (c_s+h_s)] ds,\quad  X^{t,x,c,h}_t =x.
\end{equation}
The consumption and healthcare policies $(c,h)\in\A$ describe planned expenditures depending on calendar time and past and current events, as follows. At time $0$, the household chooses deterministic policies $c_0(t)$ and $h_0(t)$, and mortality $M^{0,m,h_0}$  evolves accordingly as in \eqref{M with h}. Upon the first death at time $\tau_{1}$, the surviving household carries on with wealth $\zeta X^{0,x,c_0,h_0}_{\tau_{1}}$ and mortality $M^{0,m,h_0}_{\tau_{1}}$, switching to the deterministic policies $c_1(t)$ and $h_1(t)$. In general, if $n$ deaths have occurred by time $t$, the spending policies are $c_n(t)$ and $h_n(t)$.

Healthcare makes mortality partially endogenous, and its effect is summarized by the efficacy function $g$, with which the houehold reduces the growth of mortality $M^{0,m,h}$ by selecting appropriate $\{h_n\}\in\fL$. In the absence of healthcare (i.e. $g\equiv 0$), $M^{0,m}_t =m e^{\beta t}$ follows Gompertz' law with parameter $\beta$.
The household's objective at time $0$ is to maximize expected utility from intertemporal consumption
\begin{equation}\label{problem}
V(x,m):= \sup_{(c,h)\in \A}\E\left[\int_0^\infty e^{-\delta t}U\left(c_t\zeta^{N_t} {X}^{0,x,c,h}_t\right)dt\right],
\end{equation}
where the utility function $U$ is of the isoelastic form in \eqref{eq:power_utility}.
\nada{
\begin{equation}\label{power U}
U(z) := \frac{z^{1-\gamma}}{1-\gamma}\quad z\ge 0
\qquad\text{for}\qquad
1\ne \gamma > 0
,
\end{equation}
}


\section{Main Results}\label{sec:main results}

The optimization problem considered in this paper departs from the classical consumption-investment problem in two aspects: \emph{aging}, whereby natural mortality follows Gompertz' law, and \emph{healthcare}, which slows down mortality growth. 

To understand the separate effects of each aspect, henceforth we present the main results in order of complexity: First (Section 3.1), with neither aging nor healthcare -- a minor variation of the classical setting. Second (Section 3.2), with aging but without healthcare -- a partially new setting that provides a reference for the general model with both aging and healthcare (Section 3.3).

\subsection{Neither Aging nor Healthcare ($\beta =0$ and $g \equiv 0$)}\label{subsec:beta=0 g=0}
When the mortality rate is constant ($\beta=0$), and healthcare is unavailable ($g \equiv 0$), the household is essentially \emph{forever young}. The arrival times of death $\{\tau_n\}_{n\in\N}$ are simply 
\[
\tau_0 = 0,\quad \tau_{n+1} = \inf\{t\ge \tau_{n} \mid (t-\tau_{n})\cdot m\ge Z_{n+1}\}\ \ \forall n\ge 0,
\]
where $m\ge 0$ is the constant mortality rate.
The counting process $N$ defined in \eqref{counting process} is a Poisson process with intensity $m\ge 0$, and the collection of controls $\A$ in \eqref{A} reduces to 
\begin{equation}\label{Ac}
\cC:= \left\{c_t = \sum_{n=0}^\infty c_n(t) 1_{\{\tau_{n}\le t<\tau_{n+1}\}}\ \middle|\ \{c_n\}\in\mathfrak L\right\}.
\end{equation}
The value function in \eqref{problem} is then
\begin{equation}\label{problem_0}
V(x,m)= \sup_{c\in \cC}\E\left[\int_0^\infty e^{-\delta t}U\left(c_t\zeta^{N_t} {X}^{0,x,c}_t\right)dt\right].
\end{equation}

The next proposition describes the optimal policy in this basic setting:
\begin{proposition}\label{prop:main 1}
Let $m\ge 0$ satisfy
\begin{equation}\label{assumption_0}
\delta+(1-\zeta^{1-\gamma})m-(1-\gamma) r >0.
\end{equation}
Then, for all $x\ge 0$, 
$
V(x,m)= \frac{x^{1-\gamma}}{1-\gamma} c_0(m)^{-\gamma},
$
where
\begin{equation}\label{c_0}
c_0(m) := \frac{\delta+(1-\zeta^{1-\gamma})m}{\gamma}+\left(1-\frac1\gamma\right) r.
\end{equation}
Furthermore, $\hat c_t :=  c_0(m)$, for all $t\ge 0$, is an optimal control of \eqref{problem_0}.
\end{proposition}
\begin{proof}
See Section~\ref{subsec:proof beta=0 g=0}.
\end{proof}

First, note that the parametric restriction in \eqref{assumption_0} is a well-posedness condition, which requires that the time-preference rate is large enough to prevent consumption from being deferred indefinitely. In fact, this restriction is equivalent to a positive consumption rate \eqref{c_0}.

The main message of this proposition is that, in the absence of both aging and healthcare, the optimal consumption rate is a constant proportion of wealth, resulting from a weighted sum of the interest rate $r$ and of the discount rate $\delta+(1-\zeta^{1-\gamma})m$, which captures the effects of time-preference $\delta$ and of mortality $m$, weighted for its impact via $\zeta$. In particular, for a total loss ($\zeta=0$), mortality adds one-to-one to time-preference $\delta$, and therefore it is equivalent to a higher $\delta$, as in \cite{Yaari1965}.

Importantly, higher mortality implies a higher consumption rate for EIS $1/\gamma>1$, while the opposite holds for $1/\gamma<1$. This dependence is explained in terms of the usual income and substitution effects in response of negative wealth shocks. On one hand, higher mortality rate spurs the household to consume before wealth is reduced by deaths (substitution effect). On the other hand, mortality shocks mean less future consumption, which in turn encourages savings to alleviate the consumption shock (income effect). Either of these countervailing effects prevails above or below  $\gamma=1$. At this threshold, which corresponds to logarithmic utility, the two effects perfecly offset each other, and the consumption rate reduces to the time preference $\delta$, regardless of mortality $m$, its impact $\zeta$, and the safe rate $r$.

\subsection{Aging without Healthcare ($g\equiv 0$)}\label{subsec:beta>0 g=0}

The next conceptual step is to add aging to the optimization problem, assuming that mortality grows according to Gompertz' law ($\beta>0$), with no healthcare available ($g \equiv 0$). Thus, for an initial mortality $m\ge 0$, $M_t = me^{\beta t}$ for all $t\ge 0$. The times of death $\{\tau_{n}\}_{n\in\N}$ then become
\begin{equation}\label{tau's beta>0}
\tau_0=0,\quad \tau_{n+1}=\inf\left\{t\ge \tau_{n}\ \middle|\ (e^{\beta t} - e^{\beta\tau_{n}}) \frac{m}{\beta}\ge Z_{n+1}\right\}\ \forall n\ge 0
, 
\end{equation}
while the set of controls $\A$ in \eqref{A} again reduces to $\cC$ in \eqref{Ac}, and the value function in \eqref{problem} to the form \eqref{problem_0}.

The next proposition describes the effect of aging on the optimal consumption-savings problem:
\begin{proposition}\label{prop:main 2}
Assume either one of the two conditions: (i) $\gamma, \zeta \in (0,1)\ \hbox{and}\ \delta + (\gamma-1) r >0$; (ii)  $\gamma, \zeta>1$.
Then, for any $(x,m)\in\R^2_+$, 
$
V(x,m) = \frac{x^{1-\gamma}}{1-\gamma} u_0(m)^{-\gamma},
$
where
\begin{equation}\label{u0}
u_0(m):= \left[\frac{1}{\beta}\int_0^\infty e^{-\frac{(1-\zeta^{1-\gamma})my}{\beta\gamma}} (y+1)^{-\left(1+ \frac{\delta+(\gamma-1)r}{\beta\gamma}\right)} dy\right]^{-1}> 0
\end{equation}
is a strictly increasing function on $(0,\infty)$ satisfying 
\begin{itemize}
\item [(a)] $u_0(0) = c_0(0) = \frac{\delta+(\gamma-1)r}{\gamma}>0$,  
$\lim_{m\to \infty} \left( u_0(m)- ( c_0(m)+\beta) \right) = 0$, and
\begin{align}
\label{u0 bounds}
 c_0(m)< u_0(m) <  c_0(m)+\beta\quad \hbox{for all}\ m\in(0,\infty). 
\end{align}
\item [(b)] 
$u_0'(0+) = \infty$, $u_0'(\infty) = \frac{1-\zeta^{1-\gamma}}{\gamma}$.
\end{itemize}
Furthermore, $\hat c_t := u_0(me^{\beta t})$, for all $t\ge 0$, is an optimal control of \eqref{problem_0}. 
\end{proposition} 
\begin{proof}
See Section~\ref{subsec:proof beta>0 g=0}.
\end{proof}

\begin{figure}[t]
\label{fig:imm_for_gom}
\centering
\includegraphics[width=0.9\textwidth]{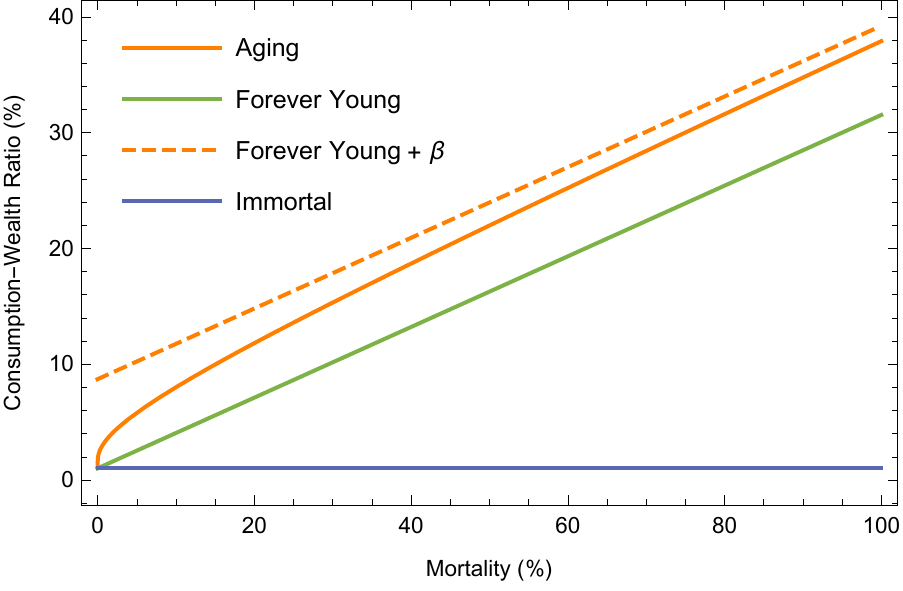}
\caption{ 
Consumption-wealth ratios (vertical, in percent) against instantaneous mortality rate (horizontal, in percent), with zero, constant, and exponentially growing mortality (solid, from bottom to top). The dashed line represents the asymptotic linear consumption rate with high, exponentially-growing mortality, which is also an upper bound at any mortality rate.
Parameters as in Section 4.
}
\end{figure}

Although this setting is known in the actuarial literature (see for example \cite{huang2012optimal}), the above result presents a few novel aspects, starting from the more complicated parametric restrictions for well-posedness, which require that either $\gamma, \zeta \le1$ or $\gamma, \zeta>1$. In fact, the growing mortality rate generates a further motive for indefinite deferral of consumption, and a resulting ill-posed problem. Indeed, if $\gamma>1$ but $\zeta<1$, with a growing mortality rate the household anticipates so much future misery that it would attempt to reduce current consumption to zero. Indeed, only a hypothetically positive mortality shock ($\zeta>1$) would lead to a well-posed problem. Such a case $\zeta>1$, which would correspond to a life-insurance policy above the value of future income, is not realistic and not pursued further: it is only discussed here to point out the source of ill-posedness for $\gamma>1$ with aging. As a result, the only economically relevant setting with aging corresponds to a EIS $1/\gamma>1$.

The optimal consumption rate in \eqref{u0} explicitly depends on mortality, though it does not admit a closed-form expression. Yet, part (a) in the Proposition makes clear comparisons to the consumption rate derived before in the case of constant mortality. In particular, for $\gamma<1$ aging implies always a higher consumption rate, but not higher by more than the growth of aging itself, and this upper bound is asymptotically reached in the old age limit, as mortality increases indefinitely. In other words, the mortality growth rate $\beta$ is the maximal increase in the consumption-wealth ratio resulting from aging, compared to another household with the same mortality but without aging. Of course, the increase in consumption rates results from the dominant substitution effect, which responds to higher mortality with earlier consumption.

Finally, part (b) in the Proposition establishes that the consumption rate with respect to mortality increases very steeply near immortality ($m=0$), while becoming asymptotically linear in the old age limit ($m=\infty$), and reaching the same slope $\frac{1-\zeta^{1-\gamma}}{\gamma}$ as in the case without aging.

Figure \ref{fig:imm_for_gom} summarizes the attributes of the settings discussed so far: the flat bottom line identifies the constant consumption rate for a household with zero mortality. The line intersecting at $m=0$ describes the consumption of a ``forever young'' household, for which mortality remains fixed at $m$ over time. The curve above this line, which also intersects at $m=0$, plots the consumption rate of an aging household, which is higher in view of the substitution effect. The top dashed line, parallel to the constant mortality line, describes the asymptotic consumption rate of the aging household for large $m$, which is also an upper bound.

Importantly, these simplified settings provide a range in which the consumption rate of the full model should lie. As healthcare curbs mortality growth, implied consumption rate should lie below the solution with aging but without healthcare, in anticipation of slower mortality growth. At the same time, consumption should be higher than with a constant mortality, at least if healthcare cannot reverse aging, as it does not in reality.


\subsection{Aging with Healthcare}\label{subsec:beta>0 g>0}

The full model incorporates both aging, with mortality increasing naturally with Gompertz' law ($\beta>0$), and healthcare ($g \ge 0$), which can slow down its growth. 

The first theorem considers the case of a general efficacy function $g$ that satisfies the well-posedness condition \eqref{g<beta} below, which stipulates that even arbitrarily high amounts of healthcare cannot arrest or reverse mortality growth, in addition to the following regularity conditions:
\begin{assumption}\label{ass:gfunct}
Let $g:\mathbb R_+\to \mathbb R_+$ be twice-differentiable with $g(0)=0$, $g'(h)>0$ and $g''(h)<0$ for $h>0$, and satisfy the Inada condition: 
\begin{equation}\label{Inada}
g'(0+)=\infty\quad \hbox{and}\quad g'(\infty)=0
.
\end{equation}
\end{assumption}
The restrictions on the other parameter values ensure well-posedness by excluding indefinite deferral of consumption.

\begin{theorem}\label{thm:main 3}
Let Assumption \ref{ass:gfunct} hold, and let $0 < \gamma < 1$ and $\bar c:= \frac{\delta}{\gamma} + \big(1-\frac1\gamma\big) r >0$. If 
\begin{equation}\label{g<beta}
g\left(I\left(\frac{1-\gamma}{\gamma}\right)\right) < \beta\quad \hbox{with}\quad I := (g')^{-1},
\end{equation}
then the value function in \eqref{problem} satisfies
$
V(x,m)=\frac{x^{1-\gamma}}{1-\gamma} u^*(m)^{-\gamma}
$
where $u^*:\R_+\to\R_+$ is the unique nonnegative, strictly increasing solution to the equation
\begin{equation}\label{HJB_u_2'}
\mathcal{L}u(m) := 
u^2(m) - c_0(m)u(m)+ m u'(m) \left(\sup\limits_{h\ge 0}\left\{g(h) - \frac{1-\gamma}{\gamma} \frac{u(m)}{mu'(m)} h\right\}-\beta\right) = 0.
\end{equation}
Furthermore, $u^*$ is strictly concave, and $(\hat c,\hat h)$ defined by
\begin{equation}\label{eq:firstorder}
\hat c_t := u^*(M_t)\quad \hbox{and} \quad \hat h_t := I\left(\frac{1-\gamma}{\gamma}\frac{u^*(M_t)}{M_t\cdot (u^*)'(M_t)}\right),\quad \hbox{for all}\ t\ge 0,
\end{equation}
optimizes \eqref{problem}. 
\end{theorem}

\begin{proof}
This result is a consequence of Proposition~\ref{prop:uniqueness} and Corollary \ref{coro:indep. of p,q} below.
\end{proof}

This result identifies the optimal consumption policy as the solution of \eqref{HJB_u_2'}, a first-order, nonlinear ODE, in which the effect of healthcare is captured by the last nonlinear term. Such a solution does not have a closed-form expression even in relatively simple settings, such as the one discussed next in Corollary~\ref{coro:g isoe}, but it is nonetheless straightforward to calculate numerically, and so are its quantitative implications.

A delicate point is that the solution to equation \eqref{HJB_u_2'} is uniquely identified without any additional boundary conditions, because the equation has only one increasing solution defined for all $m\in\mathbb R_+$, while all others explode for finite $m$ or start decreasing for $m$ large enough. Note  that natural boundary condition $u(0) = c_0(0)$ holds for any local solution in $[0,\varepsilon)$, and therefore does not identify the one defined for all $m\in\mathbb R_+$.

Also, the above result establishes the optimality condition for healthcare expenditure in \eqref{eq:firstorder}, whereby the marginal efficacy of optimal healthcare is inversely proportional to $m u'(m)/u(m)$, the elasticity of consumption with respect to mortality, where the constant of proportionality depends on preferences.

The next results provides a deeper insight on the impact of healthcare on the growth rate of mortality:
\begin{theorem}\label{thm:main 4}
Let Assumption \ref{ass:gfunct} and condition \eqref{g<beta} hold, and let $0 < \gamma < 1$, $\frac{\delta}{\gamma} + \big(1-\frac1\gamma\big) r >0$. Define
\begin{equation}\label{betag}
\beta_g := \beta - \sup\limits_{h\ge 0}\left\{g(h) - \frac{1-\gamma}{\gamma} h\right\}\in (0,\beta).
\end{equation}
As \eqref{u0} defines $u_0(m)$, define $u^g_0(m)$ analogously with $\beta_g$ in place of $\beta$.
Then, for any $m>0$,
\begin{equation}
u^g_0(m) \le u^*(m)\le \min\{u_0(m), c_0(m)+\beta_g\}
\end{equation}
and
\begin{equation}
\lim_{m\rightarrow\infty} ( c_0(m) - u^*(m) ) =  \beta_g.
\end{equation}
\end{theorem}

\begin{proof}
See Section~\ref{subsec:proof beta>0 g>0}.
\end{proof}

The message of this result is that even if the exact effect of healthcare on consumption is complicated, it does admit simple upper and lower bounds (Figure \ref{fig:u* asymptotics}). The lower bound is the one obtained from $u_0^g(m)$, the consumption rate in a model in which healthcare is not available, but mortality grows at the lower rate $\beta_g<\beta$. Indeed, the household would gladly give up access to healthcare in exchange for such a lower rate of mortality growth.

Upper bounds are consumption rates under the same mortality growth rate $\beta$ but with no access to healthcare (i.e. $u_0(m)$), and for a forever young household, augmented by the adjusted growth rate $\beta_g$ (i.e. $c_0(m)+\beta_g$). The former estimate is more accurate at younger ages, while the latter at older ages. In all cases, the minimum consumption rate, i.e. the lower bound $u_0^g(m)$, yields the sharpest estimate.

\begin{figure}[t]
\centering
\includegraphics[width=0.9\textwidth]{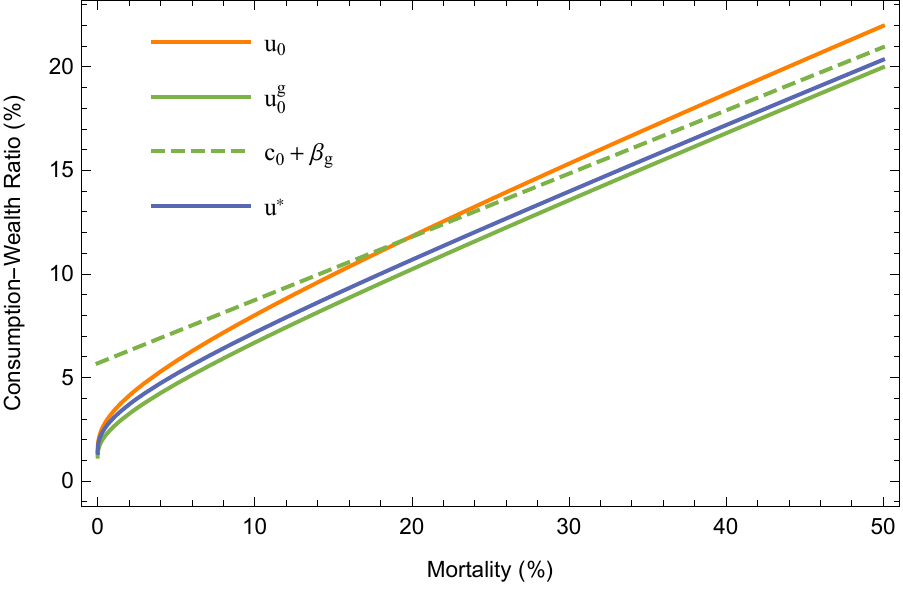}
\caption{ 
\label{fig:u* asymptotics}
Consumption-wealth ratios (vertical axis, in percent) against instantaneous mortality rates (horizontal). Solid curves correspond to aging households with natural mortality growth $\beta$, with (center) and without (top) healthcare, and with adjusted mortality growth $\beta_g$ (bottom). The dashed line represents the asymptotic linear consumption ratio, which is also an upper bound.}
\end{figure}


The next result specializes the analysis to a concrete model, assuming an efficacy function of isoelastic type, thereby enabling comparative statics and parameter estimation.

\begin{corollary}\label{coro:g isoe}
Let $0 < \gamma < 1$, $\bar c:= \frac{\delta}{\gamma} + \big(1-\frac1\gamma\big) r >0$, and $g:\R_+\to\R_+$ of the form
\[
g(z) := a \frac{z^{q}}{q},  
\]
for some $a> 0$ and $q\in(0,1)$. If 
$\frac{a^{\frac{1}{1-q}}}{q}\left(\frac{1-\gamma}{\gamma}\right)^{\frac{-q}{1-q}} < \beta$,
then
$
V(x,m)=\frac{x^{1-\gamma}}{1-\gamma} u^*(m)^{-\gamma} \hbox{for all}\ (x,m)\in\R^2_+, 
$
where $u^*:\R_+\to\R_+$ is the unique nonnegative, strictly increasing solution to the equation
\begin{equation}\label{eq:odepower}
u^2(m) - c_0(m)u(m) - \beta m u'(m) + \frac{1-q}{q} a^{\frac{1}{1-q}} \left(\frac{1-\gamma}{\gamma} u(m)\right)^{\frac{-q}{1-q}} \left(mu'(m)\right)^{\frac{1}{1-q}}=0,
\end{equation}
Furthermore, $(\hat c,\hat h)$ defined by
\[
\hat c_t := u^*(M_t)\quad \hbox{and} \quad \hat h_t :=  a^{\frac{1}{1-q}} \left(\frac{1-\gamma}{\gamma}\frac{u^*(M_t)}{M_t\cdot (u^*)'(M_t)}\right)^{\frac{-1}{1-q}},\quad \hbox{for all}\ t\ge 0,
\]
is an optimal control of \eqref{problem}. 
\end{corollary}


With this result at hand, we investigate in Section~\ref{sec:calibration} the model's implications for the optimal policies and their resulting endogenous mortality.

\section{Risky Assets}\label{sec:risky}

The main model in the paper assumes that households' savings are confined to a safe investment.
This section discusses 
extending the model to include risky assets. 
The main question is to what extent healthcare and endogenous mortality are sensitive to risky investments and -- conversely -- whether portfolio allocation is sensitive to the mortality rate. Theorem~\ref{thm:additive} below argues that the presence of risky assets with constant investment opportunities is equivalent to an increase in the safe rate, and that the resulting optimal portfolio is independent of mortality (cf. equations \eqref{eq:odepower_risky asset} and \eqref{bar c_0} below). Thus, the main results in the paper remain valid with the addition of risky assets, up to reinterpreting the safe rate parameter as an \emph{equivalent safe rate} that accounts for additional investment opportunities.

Specifically, consider a risky asset $S$ satisfying the dynamics
\begin{equation}\label{stock}
dS_t = S_t (\mu+r) dt + S_t\sigma dW_t,\quad t\ge 0,
\end{equation}
where $\mu\in\R$ and $\sigma>0$ are constants and $W$ is a standard Brownian motion independent of $\{Z_n\}_{n\in\N}$. The independence assumption is appropriate for most individuals, as death is unlikely to affect, or result from, price changes of publicly traded securities. 
Denoting by $\pi_t$ the fraction of wealth that the household invests in the risky asset at time $t$, when no deaths have occurred the wealth process is
\begin{align}\label{X}
\frac{dX_t}{X_t} &= [(1-\pi_t) r -c_t -h_t] dt + \pi_t \frac{d S_t}{ S_t}= [r+\mu\pi_t-c_t-h_t] dt + \sigma\pi_t dW_t,\quad t\ge 0.
\end{align}
To properly define the value function and derive the associated HJB equation, our probability space $(\Omega,\F,\P)$ in Section~\ref{subsec:setup} needs to be enlarged to account for the additional Brownian motion $W$. Specifically, let  $(\Omega_1,\F_1,\P_1)$ be the probability space supporting the i.i.d. exponential random variables $\{Z_n\}_{n\in\N}$, as specified in Section~\ref{subsec:setup}. Let  $(\Omega_2,\F_2,\P_2)$ be another probability space that supports the Brownian motion $W$. Then, we take $(\Omega,\F,\P)$ to be the product probability space of $(\Omega_1,\F_1,\P_1)$ and $(\Omega_2,\F_2,\P_2)$, endowed with the filtration $\{\F_t\}_{t\ge 0}$ generated by $\{Z_n\}_{n\in\N}$ and $W$. We denote by $\E_1$, $\E_2$, and $\E$ the expectations taken under $\P_1$, $\P_2$, and $\P$, respectively. 

Between two consecutive death times, $c_t$ and $h_t$ are no longer deterministic as in \eqref{A}, because they may depend on the evolution of $S$ in \eqref{stock} (or, the Brownian motion $W$). More precisely, let $L^{1}_\text{loc}(\Omega_2)$ denote the collection of processes $f:\R_+\times \Omega_2\to\R$ such that $\E_2[\int_{0}^t f(s) ds]<\infty$ for all $t\ge 0$. Also consider $L^{1,+}_\text{loc}(\Omega_2)$, the subspace of $L^{1}_\text{loc}(\Omega_2)$ containing nonnegative processes. Define 
\[
\mathfrak L' := \left\{\{f_n\}_{n\ge 0}\ \middle|\  f_0\in  L^{1}_\text{loc}(\Omega_2),\ f_n = f(Z_1,...,Z_n)\ \hbox{for some Borel}\ f:\R^n_+\to L^{1}_\text{loc}(\Omega_2)\right\}.
\]
We also consider $\mathfrak L'_{+}$, defined as $\mathfrak L'$ with $L^{1}_\text{loc}(\Omega_2)$ replaced by $L^{1,+}_\text{loc}(\Omega_2)$. Now, for any $\{h_n\}\in \mathfrak L'_+$, the death times in \eqref{tau's} can be formulated as follows: for each $\omega=(\omega_1,\omega_2)\in\Omega$,  
\begin{equation*}
\tau_{n+1}(\omega) := \inf\left\{t\ge \tau_{n}(\omega)\ \middle|\ \int_{\tau_{n}}^t M^{\tau_{{n}},m_{n},h_{n}}_s(\omega_2) ds \ge Z_{n+1}(\omega_1)\right\},\quad m_{n+1}:= M^{\tau_{n},m_{n},h_{n}}_{\tau_{{n+1}}}(\omega).
\end{equation*}
Thus, \eqref{tau>t} can be rewritten as, for any fixed $\omega_2\in\Omega_2$, 
\begin{align}\label{tau>t'}
\P_1&\left(t\in [\tau_{n}(\cdot,\omega_2),\tau_{{n+1}}(\cdot,\omega_2))\ \middle|\ Z_1,...,Z_n\right) (\omega_1)\notag \\
&= \exp\left(-\int_{\tau_{n}(\omega_1,\omega_2)}^{t} M^{\tau_{n},m_n,h_n}_s(\omega_2) ds \right)1_{\{t\ge \tau_{n}\}}(\omega_1,\omega_2)
,\quad \forall n\ge 0.
\end{align}
Similarly to \eqref{A}, the collection $\A'$ of admissible controls contains all processes $(c_t, h_t,\pi_t)$ where \begin{equation}\label{A'}
c_t = \sum_{n=0}^\infty c_n(t) 1_{\{\tau_{n}\le t<\tau_{n+1}\}},\ h_t = \sum_{n=0}^\infty h_n(t) 1_{\{\tau_{n}\le t<\tau_{n+1}\}},\ \pi_t = \sum_{n=0}^\infty \pi_n(t) 1_{\{\tau_{n}\le t<\tau_{n+1}\}},
\end{equation}
with $\{c_n\}, \{h_n\}\in\mathfrak L'_+$ and $\{\pi_n\}\in\mathfrak L'$.

Now, define the value function as
\begin{equation}\label{problem_0_risky asset}
V(x,m) := \sup_{c,h,\pi} \E\left[\int_0^\infty e^{-\delta t}U(\zeta^{N_t} X_t c_t) dt \right] = \sup_{c,h,\pi} \sum_{n=0}^\infty\E\left[\int_{\tau_{n}}^{\tau_{n+1}} e^{-\delta t}U(\zeta^{n} X_t c_t) dt\right].
\end{equation}

\subsection{Derivation of the HJB Equation and Optimal Policies}\label{subsec:heuristics}
In the following, we will first derive {\it heuristically} the HJB equation for $V(x,m)$ and the candidate optimal policies. These heuristic guesses turn out to be truly optimal, as verified in Theorem~\ref{thm:additive}.  

By the definition of $V(x,m)$ in \eqref{problem_0_risky asset},
\begin{align*}
V(x,m) &= \sup_{c,h,\pi} \E\left[\int_0^{\tau_1} e^{-\delta t} U(c_t X_t) dt +e^{-\delta\tau_1}\int_{\tau_1}^\infty e^{-\delta(t-\tau_1)}U(\zeta^{N_t} c_t X_t) dt\right]\\
&= \sup_{c,h,\pi} \E\left[\int_0^{\tau_1} e^{-\delta t} U(c_t X_t) dt +e^{-\delta\tau_1}V(\zeta X_{\tau_1}, M_{\tau_1})\right],
\end{align*}
where in the second line we assume heuristically that a dynamic programming principle for \eqref{problem_0_risky asset} holds.
By Fubini's theorem and $\P_1[\tau_1(\cdot,\omega_2)>t] = e^{-\int_0^t M^{0,m,h_0}_s(\omega_2) ds}$ from \eqref{tau>t'}, the above equation yields
\begin{align*}
V(x,m) &=  \sup_{c,h,\pi} \E_2\left[\int_0^\infty e^{-\int_0^t M_s ds} e^{-\delta t} U(c_t X_t) dt +\int_0^\infty M_t e^{-\int_0^t M_s ds} e^{-\delta t}V(\zeta X_{t}, M_{t})dt\right]\\
&= \sup_{c,h,\pi} \E_2\left[\int_0^\infty e^{-\int_0^t (\delta + M_s) ds} [U(c_t X_t)+M_t V(\zeta X_{t}, M_{t})] dt \right].
\end{align*}
This shows that the value function can be viewed alternatively as an infinite-horizon problem with running payoff $U(c_t X_t)+M_t V(\zeta X_{t}, M_{t})$ and discount rate $\delta + M_t$. Suppose that a dynamic programming principle holds for this alternative formulation, i.e.
\begin{align}\label{DPP}
V(x,m) = \sup_{c,h,\pi} \E_2\bigg[\int_0^T  e^{-\int_0^t (\delta + M_s) ds} [U&(c_t X_t) +M_t V(\zeta X_{t}, M_{t})] dt\nonumber\\
 &+ e^{-\int_0^T (\delta + M_s) ds} V(X_T,M_T)\bigg],\quad \forall T>0.
\end{align}
Then, in view of 
\begin{align*}
d\left(e^{-\int_0^t (\delta + M_s) ds} V(X_t,M_t)\right) = e^{-\int_0^t (\delta + M_s) ds} \bigg[-&(M_t+\delta) V(X_t,M_t) + V_x(X_t,M_t) dX_t\\
& + V_m(X_t,M_t) dM_t +\frac12 V_{xx}(X_t,M_t) (dX_t)^2\bigg],
\end{align*}
\eqref{DPP} implies that for all $T>0$,
\[
0 = \sup_{c,h,\pi} \E_2\bigg[\int_0^T  e^{-\int_0^t (\delta + M_s) ds} [U(c_t X_t) +M_t V(\zeta X_{t}, M_{t})+K(X_t,M_t,c_t,h_t,\pi_t)] dt\bigg],
\]
where
\begin{align*}
K(x,m,c,h,\pi) := &-(\delta+m)V(x,m)+[r+\mu\pi-c-h]xV_x(x,m)\\
&+(\beta-g(h))mV_m(x,m)+\frac12\sigma^2\pi^2x^2V_{xx}(x,m),
\end{align*}
thereby leading to the HJB equation 
\begin{align}\label{HJB_risky asset}
0\ =\ &\sup_{c\ge 0}\{U(cx)-cxV_x(x,m)\}+m V(\zeta x,m)-(\delta+m)V(x,m)+rxV_x(x,m)\nonumber\\
& +\sup_{\pi\in\R}\left\{\mu\pi x V_x(x,m)+\frac12\sigma^2\pi^2 x^2 V_{xx}(x,m)\right\}\\
&+\beta m V_m(x,m) +\sup_{h\ge 0}\{-g(h)m V_m(x,m)-hxV_x(x,m)\}.\nonumber 
\end{align}
Assuming heuristically that $V_{xx}<0$ and $V_m<0$, the above equation suggests the following candidate optimal policies $(\hat c,\hat\pi, \hat h)$ from the first-order conditions:
\begin{equation}\label{optimal policies}
\hat c =  \frac{V_x(x,m)^{-1/\gamma}}{x},\quad \hat \pi = -\frac{\mu}{\sigma^2 x}\frac{V_x(x,m)}{V_{xx}(x,m)},\quad \hat h = (g')^{-1}\left(-\frac{x V_x(x,m)}{m V_m(x,m)}\right), 
\end{equation}
which in turn imply that \eqref{HJB_risky asset} can be simplified as
\begin{align}\label{HJB_risky asset'}
0=\ &\frac{\gamma}{1-\gamma}(V_x(x,m))^{\frac{\gamma-1}{\gamma}}+mV(\zeta x,m)-(\delta+m)V(x,m)+rxV_x(x,m)+\beta m V_m(x,m)\nonumber\\
& -\frac{1}{2}\left(\frac{\mu}{\sigma}\right)^2\frac{V_{x}^2(x,m)}{V_{xx}(x,m)}-mV_m(x,m)\sup_{h\ge 0}\left\{g(h)+\frac{hxV_x(x,m)}{mV_m(x,m)}\right\}.
\end{align}
Using the ansatz $V(x,m) = \frac{x^{1-\gamma}}{1-\gamma} u(m)^{-\gamma}$, the above  equation reduces to
\begin{equation}\label{eq:odepower_risky asset}
0 = u^2(m)-\bar c_0 u(m) -\beta m u'(m) + mu'(m)\sup_{h\ge 0}\left\{g(h)-\frac{1-\gamma}{\gamma}\frac{u(m)}{m u'(m)} h\right\},
\end{equation}
where 
\begin{equation}\label{bar c_0}
\bar c_0(m) := \frac{\delta+(1-\zeta^{1-\gamma})m}{\gamma} + \left(1-\frac{1}{\gamma}\right) \left(r+ \frac{1}{2\gamma}\left(\frac{\mu}{\sigma}\right)^2\right).
\end{equation}

The next result shows that the heuristic derivation above does lead us to truly optimal strategies.

\begin{theorem}\label{thm:additive}
Including the risky asset $S$ as in \eqref{stock}, Proposition~\ref{prop:main 1}, Proposition~\ref{prop:main 2}, and Theorem~\ref{thm:main 3} still hold true, with the interest rate $r$ replaced by $r+ \frac{1}{2\gamma}\left(\frac{\mu}{\sigma}\right)^2$. In particular, the optimal consumption and health spending are as specified therein, with $c_0$ replaced by $\bar c_0$ in \eqref{bar c_0}, while the optimal portfolio is
\[
\hat \pi_t \equiv \frac{\mu}{\gamma\sigma^2},\quad t\ge 0.
\]
\nada{With multiple assets, the optimal portfolio is $\pi_t = \frac{1}{\gamma}\Sigma^{-1}\mu$, where $\mu$ is the vector of expected excess returns and $\Sigma$ their covariation matrix. Accordingly, in equation \eqref{bar c_0} the term $(\mu/\sigma)^2$ is replaced by $\mu \Sigma^{-1}\mu$.}
\end{theorem}

\begin{proof}
Comparing \eqref{c_0} and \eqref{bar c_0}, note that $\bar c_0$ differs from $c_0$, and thus equation \eqref{eq:odepower_risky asset} differs from \eqref{eq:odepower}, only in the interest rate: $r$ is now raised to $r+ \frac{1}{2\sigma}\left(\frac{\mu}{\sigma}\right)^2$. It follows that we can re-state all the results in Section~\ref{sec:main results}, with the raised interest rate, by using appropriate verification results Theorem~\ref{thm:verification with stock} and Proposition~\ref{prop:verification with stock}, extensions of Theorem~\ref{thm:verification} and Proposition~\ref{prop:verification} to account for the risky asset $S$.  The resulting optimal ratios in consumption and health spending are the same as in Section~\ref{sec:main results}, while the optimal ratio in investment is given by \eqref{optimal policies}:
\[
\hat \pi_t = -\frac{\mu}{\sigma^2 X_t} \frac{V_x(X_t,M_t)}{V_{xx}(X_t,M_t)} =    -\frac{\mu}{\sigma^2 X_t}\frac{X_t^{-\gamma} u(M_t)^{-\gamma}}{(-\gamma) X_t^{-\gamma-1} u(M_t)^{-\gamma}} =  \frac{\mu}{\gamma\sigma^2}.
\]
\end{proof}


\begin{figure}[t]
\centering
\includegraphics[width=0.47\textwidth]{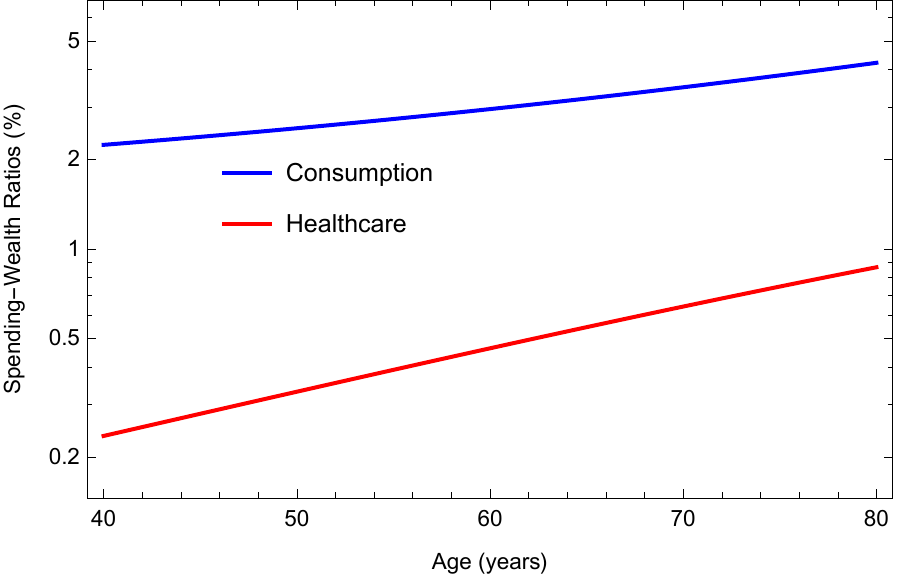}
\includegraphics[width=0.45\textwidth]{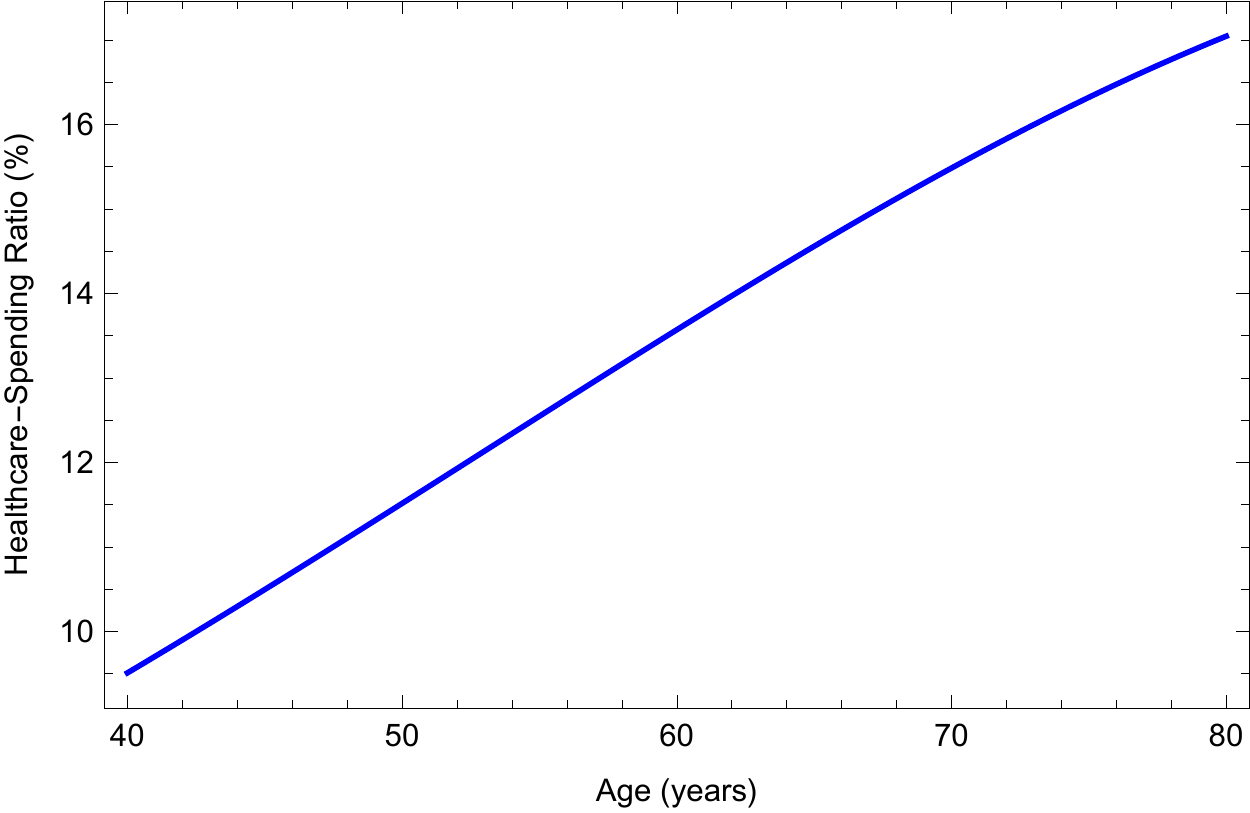}
\caption{
\label{fig:cons health}
Left: Consumption- and healthcare-wealth ratios (vertical) at adult ages (horizontal).
Right: Healthcare, as a fraction of total spending (vertical) at adult ages (horizontal).
}
\end{figure}

\section{Calibration and Implications}\label{sec:calibration}

This section discusses the implications of the isoelastic model in Corollary~\ref{coro:g isoe} for optimal policies, with the model's parameters calibrated to the values $r=1\%$, $\delta = 1\%$, $\beta=7.7\%$, $\gamma=0.67$, $\zeta=50\%$, $a=0.1$, $q=0.46$, and $m_0 = 0.019\%$.
A safe rate of $r=1\%$ approximates the long-term average real rate on Treasury bills reported by \citep{beeler2012long}, and our time preference $\delta = 1\%$ is also consistent with their estimates, while $\gamma =0.67$ corresponds to the estimates obtained by \cite{harrison2007estimating} in field experiments. 
The conventional value of $\zeta=50\%$ implies that the household loses half of its wealth with each death, and describes a household in which future income (including pensions and annuities) represents a high proportion of the net worth.

The values of the mortality natural growth $\beta$ is estimated from mortality data for the US cohort born in 1900, assuming no healthcare available. Holding these estimates constant, the healthcare parameters $a$ and $q$ appearing in the efficacy function are calibrated by matching the endogenous mortality curve with mortality data for the US cohort born in 1940.

\subsection{Healthcare Spending}

Empirical studies thoroughly confirm the familiar observation that health spending increases with  age, 
 raising the natural question of whether the model's predictions satisfy this basic property.
The left panel in Figure \ref{fig:cons health} offers an affirmative answer. At age 40, annual healthcare spending is just above 0.2\% of wealth, but it rises quickly to almost 1\% at age 80. Such increase is broadly consistent with the results of \cite{HartmanCatlinLassmanEtAl2008}, who report at age 85 and older health spending between 5.7 and 6.9 times as large as at general working-ages.

In the model, healthcare is unattractive in the young age, as mortality is unlikely, and its potential losses lead to low optimal healthcare spending. As mortality grows exponentially, healthcare spending increases accordingly, at a rate that is not far from mortality growth.

The right panel in Figure \ref{fig:cons health} compares health spending to total spending, including consumption. Although the increases in mortality implies higher spending rate in both categories, the effect is very pronounced for healthcare, as its weight increases from less than 10\% at age 40 to almost 20\% at age 80.

\begin{figure}[t]
\centering
\includegraphics[width=0.9\textwidth]{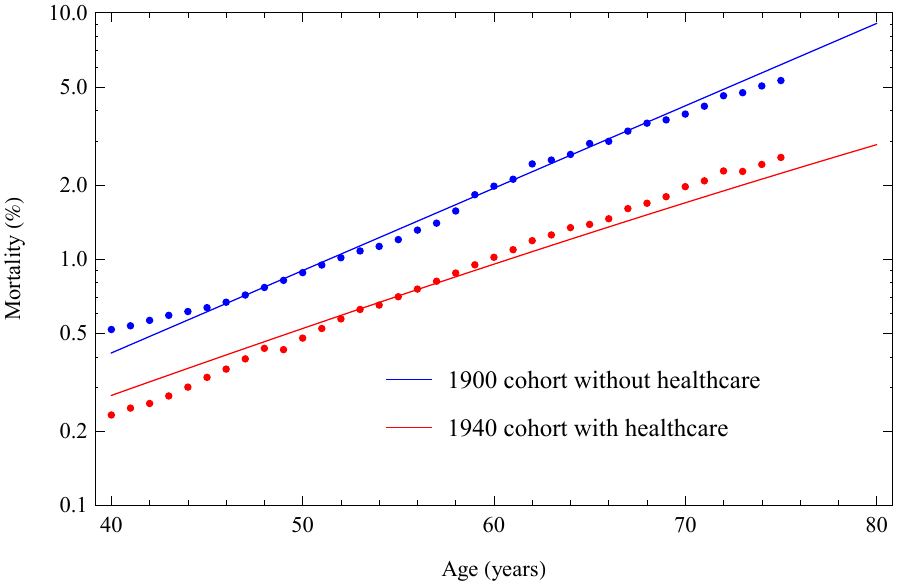}
\caption{ 
Empirical (dots) and model-implied (lines) mortality rates at adult ages for the birth cohorts of 1900 and 1940
 }
   \label{fig:mort_with_out}
\end{figure}

\subsection{Mortality}

A central question for the model at hand is to what extent it can account for the secular decrease in mortality rates. Figure \ref{fig:mort_with_out} attempts to evaluate the model performance under two simplifying assumptions. First, suppose that the cohort with birth in 1900 essentially had no access to healthcare, therefore that its mortality simply grows exponentially with Gompertz' law at the natural rate $\beta$. Second, suppose that the 1940 cohort had full access to healthcare, whence its mortality rate follows the endogenous  growth implied by the model. 

Both assumptions are clearly crude approximations, as healthcare did exist in 1900, though it was certainly less advanced and available than forty years later. Also, being forty-year old in 1940, the 1900 cohort clearly had some access to healthcare. Yet, as mortality rates in the United States are available from 1933, data on the 1900 cohort begins with 33-year olds. Similarly, at the time of writing mortality rates for 75-year-olds are not available for younger cohorts than 1940. Data availability aside, much of the rise in employer-sponsored insurance in the United States developed in the wake of wage controls enacted during World War II, which makes the 1940 cohort a reasonable choice for the our purposes.

Keeping in mind these limitations of the model and the data, Figure \ref{fig:mort_with_out} shows that the calibrated parameters are able to largely explain the decline in mortality between the two cohorts as the result of health spending. In particular, it attests the ability of the model to reproduce declines in mortality rates that are close to the ones observed historically, and that are consistent with the plausible levels of health spending described above.

\section{Conclusion}

The role of healthcare in reducing mortality rather than increasing current utility calls for a different treatment from other forms of consumption. At the same time, the exponential growth of mortality is a leading driver of health spending. This paper combines these features in a model of optimal choice of consumption and health spending, in which savings accrue a constant interest rate.

The model captures a few stylized facts on mortality and health spending. Healthcare leads to an endogenously determined mortality curve that is close to exponential, albeit with a lower growth rate, consistently with the observed secular decline in mortality. Healthcare spending continues to increase at adult ages, outpacing consumption growth, and therefore also its share of total spending increases.

\appendix

\section{Proofs of Main Results}
In this section, we prove the main results in Section~\ref{sec:main results} by verification arguments, relying on Theorem \ref{thm:verification} and \propref{prop:verification}. 
Given $f:\R_+\to\R$, consider its Legendre transform $\widetilde{f}(y):=\sup_{x\ge 0}\{f(x)-xy\}$ for $y\in\R_+$. A heuristic derivation as in Section \ref{subsec:heuristics} shows that the Hamilton-Jacobi equation associated with $V(x,m)$ in \eqref{problem} is
\begin{equation}\label{HJB}
\begin{split}
\widetilde{U}(w_x(x,m)) & +m w(\zeta x,m)-\left(\delta+m\right)w(x,m) \\
&+ rxw_x(x,m) +\beta m w_m(x,m)+\sup_{h\ge 0}\left\{ -m w_m(x,m)g(h)-h x w_x(x,m)\right\}=0.
\end{split}
\end{equation}
This is simply \eqref{HJB_risky asset'} without the second-order term, contributed by the added risky asset in Section~\ref{sec:risky}.

\subsection{Neither Aging nor Healthcare ($\beta =0$ and $g \equiv 0$)} \label{subsec:proof beta=0 g=0}
Recall the setup in Section~\ref{subsec:beta=0 g=0}. 
Since $V(x,m)$ is nondecreasing in $x$ by definition, the supremum in the last term of \eqref{HJB} vanishes, leading to the equation 
\begin{equation}
\widetilde{U}(w_x(x,m))  +m w(\zeta x,m)-\left(\delta+m\right)w(x,m) + rxw_x(x,m) =0.
\end{equation}
If $V$ is of the form $V(x,m) = \frac{x^{1-\gamma}}{1-\gamma} v(m)$. Then, the above equation reduces to 
\[
v(m)^{1-\frac1\gamma}- c_0(m) v(m)=0,
\]
where $ c_0(m)$ is defined as in \eqref{c_0}. By setting $v(m)=u(m)^{-\gamma}$, we obtain from the above equation 
\begin{equation}
u^2(m) - c_0(m)u(m)=0,
\end{equation}
whence $u(m) =  c_0(m)$. We then prove Proposition~\ref{prop:main 1} by verification.

\begin{proof}[Proof of Proposition~\ref{prop:main 1}]
Set $w(x,m):=\frac{x^{1-\gamma}}{1-\gamma}  c_0(m)^{-\gamma}$. Note that \eqref{assumption_0} implies $ c_0(m) >0$ for all $\gamma>0$, $\gamma\neq 1$. 

{\bf Case I:} $0<\gamma<1$. By \thmref{thm:verification}, it suffices to verify \eqref{t to infty} and \eqref{n to infty} under current context. For any $x\ge 0$, $c\in \cC$, and $n\in\N$, since $0<\gamma<1$, we have
\begin{align*}
0&\le \E\left[e^{-(\delta+m)(t-\tau_n)} w\left(\zeta^n X^{0,x,c}_t,m\right)\ \middle|\ Z_1,...,Z_n\right]\\
 &\le e^{-(\delta+m)(t-\tau_n)} \frac{(\zeta^n X^{0,x,c}_{\tau_n})^{1-\gamma}}{1-\gamma} e^{(1-\gamma)r(t-\tau_n)}  c_0(m)^{-\gamma}\to\  0\ \ \hbox{a.s.}\quad \hbox{as}\ t\to\infty,
\end{align*}
where the convergence follows from \eqref{assumption_0}. This already verifies \eqref{t to infty}. On the other hand, since $\tau_n$ is the sum of $n$ independent, identically distributed exponential random variables with mean $1/m$, 
\begin{equation}\label{tau_n estimate}
\begin{split}
0&\le \E\left[e^{-\delta \tau_n} w\left(\zeta^n X^{0,x,c}_{\tau_n},m\right)\right] \le \zeta^{(1-\gamma)n}\frac{x^{1-\gamma}}{1-\gamma}  c_0(m)^{-\gamma} 
\E\left[e^{-\delta \tau_n}e^{(1-\gamma) r\tau_n} \right]\\
& = \zeta^{(1-\gamma)n}\frac{x^{1-\gamma}}{1-\gamma}  c_0(m)^{-\gamma} \int_0^\infty e^{-(\delta+(\gamma-1)r) t} m^n e^{-mt} \frac{t^{n-1}}{(n-1)!}dt \\
&=  \frac{\left(m\zeta^{(1-\gamma)}\right)^n}{(n-1)!}\frac{x^{1-\gamma}}{1-\gamma}  c_0(m)^{-\gamma} \int_0^\infty e^{-(\delta+m+(\gamma-1)r) t}  t^{n-1} dt\\
&= \left(\frac{m\zeta^{(1-\gamma)}}{\delta+m+(\gamma-1)r}\right)^n \frac{x^{1-\gamma}}{1-\gamma}  c_0(m)^{-\gamma},
\end{split}
\end{equation}
where the third equality requires $\delta+m+(\gamma-1)r>0$, which is true under \eqref{assumption_0}. Finally, noting that \eqref{assumption_0} implies $\frac{m\zeta^{(1-\gamma)}}{\delta+m+(\gamma-1)r}\in(0,1)$, we conclude that $\E\big[e^{-\delta \tau^m_n} w\big(\zeta^n X^{0,x,c}_{\tau^m_n},m\big)\big]\to 0$ as $n\to\infty$, which verifies \eqref{n to infty}.

{\bf Case II:} $\gamma>1$. By \propref{prop:verification}, it suffices to establish \eqref{t to infty'}-\eqref{eps to 0} and show that $\hat c_t \equiv  c_0(m)$ satisfies \eqref{t to infty} and \eqref{n to infty}. For any $\eps>0$, since $g\equiv 0$, the sequence $\{\tau^\eps_n\}_{n\ge 0}$ constructed in Appendix~\ref{sec:appendix} coincides with $\{\tau_n\}_{n\ge 0}$. The counting process $N^\eps$ in \eqref{counting process eps} is therefore the same as $N$ in \eqref{counting process}. It follows that \eqref{eps to 0} trivially holds under current context. Given $x\ge 0$, $c\in \cC$, and $\eps>0$, consider
\begin{equation}\label{c eps}
(c_\eps)_t := \frac{c_t X^{0,x,c}_t}{X^{0,x,c}_t + \eps e^{rt}}\quad \forall t\ge 0.
\end{equation}
By construction, $X^{0,x+\eps,c_\eps}_t = X^{0,x,c_\eps} + \eps e^{rt}$ for all $t\ge 0$. This, together with $\gamma>1$, implies that for any $n\in\N$,
\[
0\ge \E\left[e^{-(\delta+m)(t-\tau_n)} w\left(\zeta^n X^{0,x+\eps,c_\eps}_t,m\right)\ \middle|\ Z_1,...,Z_n\right]\ge e^{-(\delta+m)(t-\tau_n)} \frac{(\zeta^n\eps)^{1-\gamma}e^{(1-\gamma)rt}}{1-\gamma} c_0(m)^{-\gamma}.
\]
Since the right hand side converges to 0 a.s. as $t\to\infty$, the above inequality in particular implies that \eqref{t to infty'} holds. On the other hand, by a calculation similar to \eqref{tau_n estimate},
\begin{align*}
0&\ge  \E\left[e^{-\delta \tau_n} w\left(\zeta^n X^{0,x+\eps,c_\eps}_{\tau_n},m\right)\right] \ge \zeta^{(1-\gamma)n}\frac{\eps^{1-\gamma}}{1-\gamma}  c_0(m)^{-\gamma} \E\left[e^{-\delta \tau_n}e^{(1-\gamma) r\tau_n} \right]\\
&= \left(\frac{m\zeta^{(1-\gamma)}}{\delta+m+(\gamma-1)r}\right)^n \frac{\eps^{1-\gamma}}{1-\gamma} c_0(m)^{-\gamma}.
\end{align*}
Since \eqref{assumption_0} again guarantees that $\frac{m\zeta^{(1-\gamma)}}{\delta+m+(\gamma-1)r}\in(0,1)$, we conclude that $\E[e^{-\delta \tau_n} w(\zeta^n X^{0,x,c}_{\tau_n},m)]\to 0$ as $n\to\infty$, which verifies \eqref{n to infty'}. Now, with $\hat c_t \equiv  c_0(m)$, for any $n\in\N$, 
\[
0\ge e^{-(\delta+m) (t-\tau_n)} w(X^{0,x,\hat c^{}}_{t},m) = \frac{(X^{0,x,\hat c}_{\tau_n})^{1-\gamma}}{1-\gamma}  c_0(m)^{-\gamma} e^{-[\delta+m+(\gamma-1)(r- c_0(m))](t-\tau_n)},\quad \hbox{if $t>\tau_n$}.
\]
Observing that \eqref{assumption_0} implies
\begin{equation}\label{>0}
\delta+m+(\gamma-1)(r- c_0(m))= c_0(m)+ m\zeta^{1-\gamma}>0,
\end{equation}
we conclude that $e^{-(\delta+m) (t-\tau_n)} w(X^{0,x,\hat c^{}}_t,m)\to 0$ a.s. as $t\to\infty$. This shows that $\hat c$ satisfies \eqref{t to infty}. Thanks to \eqref{>0}, a calculation similar to \eqref{tau_n estimate} yields
\begin{align*}
0&\ge  \E\left[e^{-\delta \tau_n} w\left(\zeta^n X^{0,x,\hat c^{}}_{\tau_n},m\right)\right] = \zeta^{(1-\gamma)n}\frac{x^{1-\gamma}}{1-\gamma}  c_0(m)^{-\gamma} \E\left[e^{-\delta \tau_n}e^{(1-\gamma) (r- c_0(m))\tau_n} \right]\\
&= \left(\frac{m\zeta^{1-\gamma}}{ c_0(m)+ m\zeta^{1-\gamma}}\right)^n \frac{x^{1-\gamma}}{1-\gamma}  c_0(m)^{-\gamma}\to 0\quad \hbox{as}\ n\to\infty,
\end{align*}
which shows that $\hat c$ satisfies \eqref{n to infty}. 
\end{proof}


\subsection{Aging without Healthcare ($g\equiv 0$)} \label{subsec:proof beta>0 g=0}
Recall the setup in Section~\ref{subsec:beta>0 g=0}. The Hamilton-Jacobi equation \eqref{HJB} associated with the value function $V(x,m)$ now becomes
\begin{equation}
\widetilde{U}(w_x(x,m))  +m w(\zeta x,m)-\left(\delta+m\right)w(x,m) + rxw_x(x,m)+\beta m w_m(x,m) =0.
\end{equation}
If $V$ is of the form $V(x,m) = \frac{x^{1-\gamma}}{1-\gamma} v(m)$, the above equation turns into
\[
v(m)^{1-\frac1\gamma}- c_0(m) v(m) + \frac{\beta m}{\gamma} v'(m)=0,
\]
where $ c_0(m)$ is given by \eqref{c_0}. Setting $v(m)=u(m)^{-\gamma}$, we obtain from the above equation 
\begin{equation}\label{HJB_u_1}
u^2(m) - c_0(m)u(m)-\beta m u'(m)=0,
\end{equation}
which admits the general solution
\[
u(m) = \beta e^{-\frac{(1-\zeta^{1-\gamma})m}{\beta\gamma}} \left[C \beta m^{\frac{\delta+(\gamma-1)r}{\beta\gamma}}+ \int_1^\infty e^{-\frac{(1-\zeta^{1-\gamma})mu}{\beta\gamma}} u^{-\left(1+ \frac{\delta+(\gamma-1)r}{\beta\gamma}\right)} du\right]^{-1},
\]
where $C\in\R$ is a constant to be determined. Proposition~\ref{prop:main 2} states that taking $C=0$, which turns $u(m)$ into $u_0(m)$ in \eqref{u0}, leads to our value function $V(x,m)$. In the following, we separate the proof of Proposition~\ref{prop:main 2} into two parts.

\begin{lemma}\label{lem:properties of u0}
Under the assumptions of Proposition~\ref{prop:main 2}, $u_0$ defined in \eqref{u0} is a strictly increasing function on $(0,\infty)$ satisfying (a) and (b) in Proposition~\ref{prop:main 2}. 
\end{lemma}

\begin{proof}
The definition of $u_0$ in \eqref{u0} directly implies that $u_0$ is strictly increasing, $u_0(0)= \frac{\delta+(\gamma-1)r}{\gamma}$, and $u_0'(0+)=\infty$. Since $u_0$ solves \eqref{HJB_u_1},  for any $m\in(0,\infty)$,
\[
u_0(m) -  c_0(m) = \frac{\beta m u_0'(m)}{u_0(m)} >0,
\]
where the inequality follows from $u_0$ being positive and strictly increasing. On the other hand, using $y+1 < e^y$ for $y > 0$, \eqref{u0} yields 
\begin{align*}
u_0(m) &< \beta \left[\int_0^\infty \exp\left\{-\left(\frac{(1-\zeta^{1-\gamma})m}{\beta\gamma}+1+ \frac{\delta+(\gamma-1)r}{\beta\gamma}\right) y\right\} dy\right]^{-1}\\
& = \beta \left[\frac{\delta+(1-\zeta^{1-\gamma})m+(\gamma-1)r}{\beta\gamma}+1\right] =  c_0(m) +\beta.
\end{align*}
Finally, Taylor's expansion of $u_0(m)$ at infinity shows
\[
u_0(m) =   c_0(m) +\beta + O(1/m).
\]
This implies $u_0(m) - (  c_0(m) +\beta) \to 0$ as $m\to\infty$, and 
\[
\lim_{m\to\infty} u_0'(m) = \lim_{m\to\infty} \frac{u_0(m)}{m} = \frac{1-\zeta^{1-\gamma}}{\gamma}. 
\]
\end{proof}

\begin{proof}[Proof of Proposition~\ref{prop:main 2}]
Properties (a) and (b) are established in Lemma~\ref{lem:properties of u0}. 
Here, we prove the rest of the claims in Proposition~\ref{prop:main 2}.
Set $w(x,m) := \frac{x^{1-\gamma}}{1-\gamma} u_0(m)^{-\gamma}$. By Lemma~\ref{lem:properties of u0}, it remains to show that $V(x,m) = w(x,m)$ and $\hat c_t := u_0(m e^{\beta t})$, $t\ge 0$, is an optimal control of \eqref{problem_0}. First, we observe that $\hat c$ is an element of $\cC$. Indeed, for any compact subset $K$ of $\R_+$, thanks to $u_0 \le c_0 +\beta$ in Lemma~\ref{lem:properties of u0}, 
\begin{equation}\label{c admissible}
\int_K \hat c_t dt \le \int_K \frac{\delta+(1-\zeta^{1-\gamma}) me^{\beta t}+(\gamma-1)r}{\gamma} + \beta\ dt <\infty.  
\end{equation}
Now we deal with two cases separately.

{\bf Case I:} condition (i) holds. By \thmref{thm:verification}, it suffices to verify \eqref{t to infty} and \eqref{n to infty} under current context. For any $(x,m)\in \R^2_+$, $c\in \cC$, and $n\in\N$, by using $\gamma\in(0,1)$ and $X^{0,x,c}_t \le X^{0,x,c}_{\tau_n}\exp\left(r(t-\tau_n)\right)$ on the set $\{t\ge \tau_n\}$, 
\begin{align*}
0&\le \E\left[\exp\left(-\int_{\tau_n}^t (\delta+ m e^{\beta s}) ds\right) w\left(\zeta^n X^{0,x,c}_t, me^{\beta t}\right)\ \middle|\ Z_1,...,Z_n\right]\\
 &\le e^{-(\delta+m)(t-\tau_n)} \frac{(\zeta^n X^{0,x,c}_{\tau_n})^{1-\gamma}}{1-\gamma} e^{(1-\gamma)r(t-\tau_n)} u_0(me^{\beta t})^{-\gamma}\to 0\ \ \hbox{a.s.}\quad \hbox{as}\ t\to\infty,
\end{align*}
where the convergence follows from  $\delta+ (\gamma-1) r >0$ and $u_0$ being a nondecreasing function by definition. This in particular implies \eqref{t to infty}. On the other hand,
\begin{equation*}
\begin{split}
0&\le \E\left[e^{-\delta \tau_n} w\left(\zeta^n X^{0,x,c}_{\tau_n},me^{\beta \tau_n}\right)\right] \le \zeta^{(1-\gamma)n}\frac{x^{1-\gamma}}{1-\gamma}  \E\left[e^{-\delta \tau_n}e^{(1-\gamma) r\tau_n} u_0(me^{\beta \tau_n})^{-\gamma}\right]\\
&\le \zeta^{(1-\gamma)n}\frac{x^{1-\gamma}}{1-\gamma} u_0(m)^{-\gamma} \E\left[e^{-(\delta+(\gamma-1)r) \tau_n}\right]
\le \zeta^{(1-\gamma)n}\frac{x^{1-\gamma}}{1-\gamma} u_0(m)^{-\gamma} \to 0\quad \hbox{as}\quad n\to\infty,
\end{split}
\end{equation*}
where the fourth inequality follows from $\delta+ (\gamma-1) r >0$ and the convergence is due to $\gamma,\zeta\in(0,1)$. Thus, \eqref{n to infty} is satisfied.

{\bf Case II:} condition (ii) holds. By \propref{prop:verification}, it suffices to establish \eqref{t to infty'}-\eqref{eps to 0} and show that $\hat c_t := u_0(me^{\beta t})$ satisfies \eqref{t to infty} and \eqref{n to infty}. For any $\eps>0$, since $g\equiv 0$, the sequence $\{\tau^\eps_n\}_{n\ge 0}$ constructed in Appendix~\ref{sec:appendix} coincides with $\{\tau_n\}_{n\ge 0}$. The counting process $N^\eps$ in \eqref{counting process eps} is therefore the same as $N$ in \eqref{counting process}. Thus, \eqref{eps to 0} trivially holds under current context. Given $(x,m)\in \R^2_+$, $c\in \cC$, and $\eps>0$, consider the consumption policy $c_\eps$ as in \eqref{c eps}, and the associated property $X^{0,x+\eps,c_\eps}_t = X^{0,x,c}_t + \eps e^{rt}$ for all $t\ge 0$. We then deduce from $\gamma>1$ and $u_0$ being an nondecreasing function that for any $n\in\N$,
\begin{align*}
0&\ge \E\left[\exp\left(-\int_{\tau_n}^t (\delta+ m e^{\beta s}) ds\right) w\left(\zeta^n X^{0,x+\eps,c_\eps}_t,me^{\beta t}\right)\ \middle|\ Z_1,...,Z_n\right]\\
&\ge e^{-(\delta+m)(t-\tau_n)} \frac{(\zeta^n\eps)^{1-\gamma}e^{(1-\gamma)rt}}{1-\gamma} u_0(m)^{-\gamma}\to 0\ \ \hbox{a.s.}\quad\hbox{as}\ t\to\infty,
\end{align*}
which in particular implies \eqref{t to infty'}. Since $\gamma>1$ ensures $\delta+(\gamma-1)r>0$, we have
\begin{align*}
0\ge \E\left[e^{-\delta \tau_n} w\left(\zeta^n X^{0,x+\eps ,c_\eps}_{\tau_n},me^{\beta \tau_n}\right)\right] &\ge \zeta^{(1-\gamma)n}\frac{\eps^{1-\gamma}}{1-\gamma} u_0(m)^{-\gamma} \E[e^{-(\delta+(\gamma-1)r)\tau_n}]\\
& \ge \zeta^{(1-\gamma)n}\frac{\eps^{1-\gamma}}{1-\gamma} u_0(m)^{-\gamma}\to 0\quad \hbox{as}\ n\to\infty,
\end{align*} 
where the convergence follows from $\gamma,\zeta>1$. This verifies \eqref{n to infty'}. Now, for any $n\in\N$, applying $\hat c_t := u_0(me^{\beta t})$, $t\ge 0$, yields
\begin{align}
0 &\ge \E\left[\exp\left(-\int_{\tau_n}^t (\delta+ m e^{\beta s}) ds\right) w\left(X^{0,x,\hat c}_t,me^{\beta t}\right)\ \middle|\ Z_1,...,Z_n\right]\nonumber\\
&\ge e^{-[\delta+(\gamma-1)r] (t-\tau_n)}\exp\left( -\int_{\tau_n}^t \left[m e^{\beta s}-(\gamma-1)u_0(me^{\beta s})\right] ds\right) 
 \frac{(X^{0,x,\hat c}_{\tau_n})^{1-\gamma}}{1-\gamma} u_0(m)^{-\gamma}\label{bar c estimate}
\end{align}
on the set $\{t \ge \tau_n\}$. By direct calculation, $u_0(m)$ defined in \eqref{u0} can be expressed as
\[
u_0(m) = \beta\frac{e^{-\frac{m(1-\zeta^{1-\gamma})}{\beta\gamma}}\left(\frac{m(1-\zeta^{1-\gamma})}{\beta\gamma}\right)^{-\frac{\delta+(\gamma-1)r}{\beta\gamma}}}{\overline\Gamma\left(-\frac{\delta+(\gamma-1)r}{\beta\gamma},\frac{m(1-\zeta^{1-\gamma})}{\beta\gamma}\right)},
\]
where $\overline\Gamma(s,z):=\int_z^\infty t^{s-1} e^{-t} dt$ is the upper incomplete gamma function. Recalling the property $\frac{\overline\Gamma(s,z)}{e^{-z}z^{s-1}}\to 1$ as $z\to\infty$, it follows that
\begin{equation}\label{u_0<m}
\lim_{m\to\infty} \frac{m}{\gamma u_0(m)}= \frac{1}{1-\zeta^{1-\gamma}} >1,
\end{equation}
whence $me^{\beta s} > \gamma u_0(m e^{\beta s})$ for $s$ large enough. This, together with $u_0$ being a positive nonincreasing function, shows that $\int_{\tau_n}^t \left[m e^{\beta s}-(\gamma-1)u_0(me^{\beta s})\right] ds\to\infty$ a.s. as $t\to\infty$. We then conclude from \eqref{bar c estimate} that $\hat c$ satisfies \eqref{t to infty}. It remains to show that $\hat c$ satisfies \eqref{n to infty}. Observe that
\begin{align}
0&\ge \E\left[e^{-\delta \tau_n} w\left(\zeta^n X^{0,x,\hat c}_{\tau_n},me^{\beta \tau_n}\right)\right] \nonumber\\
&\ge \zeta^{(1-\gamma)n}\frac{x^{1-\gamma}}{1-\gamma} u_0(m)^{-\gamma} \E\left[e^{-(\delta+(\gamma-1)r) \tau_n}  \exp\left(\int_0^{\tau_n} (\gamma-1)u_0(me^{\beta s})ds\right) \right]\nonumber\\
& \ge \zeta^{(1-\gamma)n}\frac{x^{1-\gamma}}{1-\gamma} u_0(m)^{-\gamma} 
\int_0^\infty  \P(\tau_n\in dt)  \exp\left(\int_0^{t} (\gamma-1)u_0(me^{\beta s})ds\right)    dt,\label{esti}
\end{align}
where the third line above follows from $e^{-(\delta+(\gamma-1)r) \tau_n}\le 1$ as $\gamma >1$. 
Note that for each $n\in\N$, $\sum_{i=1}^n Z_i$ has a gamma distribution with the law $\P(\sum_{i=1}^n Z_i\le z) = \frac{1}{\Gamma(n)} \underline\Gamma(n, z)$, where $\Gamma(s):=\int_0^\infty t^{s-1} e^{-t} dt$ is the gamma function and $\underline\Gamma(s,z):=\int_0^z t^{s-1} e^{-t} dt$ is the lower incomplete gamma function.
We then observe from \eqref{tau's beta>0} that
\[
\P(\tau_n\le t) = \P\left( \frac{m}{\beta} (e^{\beta t}-1) \ge \sum_{i=1}^n Z_i \right) = \frac{1}{\Gamma(n)} \underline\Gamma\left(n, \frac{m}{\beta}(e^{\beta t}-1)\right) \quad\ \forall t\ge 0.
\]
It follows that 
\begin{equation}\label{density beta>0}
\P(\tau_n\in dt) = \frac{d}{dt} \P(\tau_n\le t) = \frac{1}{\Gamma(n)} \left(\frac{m}{\beta}\right)^n (e^{\beta t}-1)^{n-1} e^{-\frac{m}{\beta}(e^{\beta t}-1)} \beta e^{\beta t}\quad \forall t\ge 0. 
\end{equation}
Also, for any 
\begin{equation}\label{alpha range}
\alpha\in \left((1-1/\gamma)(1-\zeta^{1-\gamma}),1\right),
\end{equation}
\eqref{u_0<m} yields $\frac{m}{\gamma u_0(m)}> \frac{\alpha}{1-\zeta^{1-\gamma}}$ for $m$ large enough. Thus, there exists $t^*>0$ such that 
\begin{equation}\label{u_0<m'}
\frac{1-\zeta^{1-\gamma}}{\alpha\gamma}m e^{\beta t}> u_0(me^{\beta t}),\quad \hbox{for}\ t\ge t^*. 
\end{equation}
By setting $C(t^*):=\exp\big(\int_0^{t^*} (\gamma-1)u_0(me^{\beta s})ds\big)$, we obtain from  \eqref{esti} that
\begin{align}
0&\ge \E\left[e^{-\delta \tau_n} w\left(\zeta^n X^{0,x,\hat c}_{\tau_n},me^{\beta \tau_n}\right)\right]\nonumber \\
&\ge \zeta^{(1-\gamma)n}\frac{x^{1-\gamma}}{1-\gamma} \frac{u_0(m)^{-\gamma}}{\Gamma(n)} \left(\frac{m}{\beta}\right)^n 
\bigg[C(t^*)\int_0^{t^*} (e^{\beta t}-1)^{n-1} e^{-\frac{m}{\beta}(e^{\beta t}-1)} \beta e^{\beta t} dt\nonumber\\ 
 &\hspace{2in}   +\int_{t^*}^\infty  (e^{\beta t}-1)^{n-1} e^{-\frac{m}{\beta}(e^{\beta t}-1)} \beta e^{\beta t}  e^{(1-\frac1\gamma)(1-\zeta^{1-\gamma})\frac{m}{\alpha\beta}(e^{\beta t}-1)} dt\bigg]\nonumber\\
&\ge   \zeta^{(1-\gamma)n}\frac{x^{1-\gamma}}{1-\gamma} \frac{u_0(m)^{-\gamma}}{\Gamma(n)} \left(\frac{m}{\beta}\right)^n 
 \bigg[C(t^*)\int_0^\infty y^{n-1} e^{-\frac{m}{\beta}y}  dy+ \int_0^\infty y^{n-1} e^{-\left[1-\frac{1}{\alpha}(1-\frac1\gamma)(1-\zeta^{1-\gamma})\right]\frac{m}{\beta}y}dy  \bigg]\nonumber\\
&=  \zeta^{(1-\gamma)n}\frac{x^{1-\gamma}}{1-\gamma}u_0(m)^{-\gamma} \left[C(t^*) + \left(1-\frac{1}{\alpha}\left(1-\frac1\gamma\right)\left(1-\zeta^{1-\gamma}\right)\right)^{-n}\right],\label{need gamma<2}
\end{align}
where the second line follows form \eqref{density beta>0} and \eqref{u_0<m'}, the fourth line is due to the change of variable $y=e^{\beta t}-1$, and the last equality holds when $1-\frac{1}{\alpha}\big(1-\frac1\gamma\big)\big(1-\zeta^{1-\gamma}\big)>0$, which is true under \eqref{alpha range}. Noting that \eqref{alpha range} implies $1-\frac1\alpha\big(1-\frac1\gamma\big)\big(1-\zeta^{1-\gamma}\big)> 1-(1-\zeta^{1-\gamma}) =\zeta^{1-\gamma}$, we conclude from \eqref{need gamma<2} that $\E\big[e^{-\delta \tau_n} w\big(\zeta^n X^{0,x,\hat c}_{\tau_n},me^{\beta \tau_n}\big)\big]\to 0$ as $n\to\infty$, i.e. $\hat c$ satisfies \eqref{n to infty}.
\end{proof}


\subsection{Aging with Healthcare}\label{subsec:proof beta>0 g>0}
Recall the setup in Section~\ref{subsec:beta>0 g>0} where mortality increases naturally according to Gompertz' law ($\beta>0$), and at the same time healthcare is available (i.e. $g:\R_+\to\R_+$ is not constantly 0) to slow down the mortality growth. 
For the rest of this section, let Assumption \ref{ass:gfunct} hold, and denote by $I:\R_+\to\R_+$ the inverse function of $g'$, and note that $I$ is strictly decreasing.

If the value function \eqref{problem} is of the form $V(x,m) = \frac{x^{1-\gamma}}{1-\gamma} v(m)$, then \eqref{HJB} yields 
\[
v(m)^{1-\frac1\gamma}- c_0(m) v(m) + \frac{\beta m}{\gamma} v'(m) + \frac{1-\gamma}{\gamma}\sup_{h\ge 0}\left\{-\frac{v'(m)}{1-\gamma}\left(m g(h)+h\frac{(1-\gamma) v(m)}{v'(m)}\right)\right\} =0,
\]
where $ c_0(m)$ is given by \eqref{c_0}. By setting $v(m)=u(m)^{-\gamma}$ and assuming that $u'(m)\ge 0$, the above equation becomes
\begin{equation}\label{HJB_u_2}
u^2(m) - c_0(m)u(m)-\beta m u'(m)+
\begin{cases}
m u'(m) \sup\limits_{h\ge 0}\left\{g(h) - \frac{1-\gamma}{\gamma} \frac{u(m)}{mu'(m)} h\right\}=0, &\hbox{if}\ 0<\gamma<1,\\
m u'(m) \inf\limits_{h\ge 0}\left\{g(h) - \frac{1-\gamma}{\gamma} \frac{u(m)}{mu'(m)} h\right\}=0, &\hbox{if}\ \gamma>1.
\end{cases}
\end{equation}
Since $g$ is a nondecreasing function with $g(0)=0$, the infimum above equals $0$. That is, when $\gamma>1$, the above equation reduces to \eqref{HJB_u_1}, and the associated value function and optimal consumption strategy are as described in Proposition~\ref{prop:main 2}. 

Hence, we focus on the case $0 < \gamma < 1$ in the rest of this section. The equation \eqref{HJB_u_2} is now $\mathcal{L}u(m) = 0$ as in \eqref{HJB_u_2'}. 
In the following, we will employ Perron's method to construct solutions to \eqref{HJB_u_2'}, under the assumption that
\begin{equation}\label{bar c}
\bar c:= \frac{\delta}{\gamma} + \left(1-\frac1\gamma\right) r >0.
\end{equation} 

\begin{definition}\label{def:Pi}
Let $\Pi$ be the collection of $(p,q)$, where $p,q:\R_+\to\R$ are continuous and satisfy
\begin{itemize}
\item [(i)] $ c_0\le p\le q\le c_0 + \beta$ on $(0,\infty)$.
\item [(ii)] $p$ and $q$ are strictly increasing and concave.
\item [(iii)] $p$ (resp. $q$) is a viscosity subsolution (resp. supersolution) to \eqref{HJB_u_2'} on $(0,\infty)$.
\end{itemize}
For any $(p,q)\in\Pi$, let $\mathcal{S}(p,q)$ denote the collection of continuous $f:\R_+\to \R$ such that 
\begin{enumerate}
\item [1.] $p\le f\le q$ on $(0,\infty)$.
\item [2.] $f$ is strictly increasing and concave.
\item [3.] $f$ is a  viscosity supersolution to \eqref{HJB_u_2'} on $(0,\infty)$.
\end{enumerate}
\end{definition}

\begin{remark}\label{rem:not empty}
Under \eqref{bar c} and \eqref{g<beta}, $\Pi\neq\emptyset$. Indeed, $c_0+\beta$ is a supersolution to \eqref{HJB_u_1}, and thus a supersolution to \eqref{HJB_u_2'}. Specifically, for any $m>0$,
\begin{equation}\label{c0+beta supersolution}
\mathcal{L}(c_0+\beta)(m)\ge (c_0(m)+\beta)\beta-\beta m \left(\frac{1-\zeta^{1-\gamma}}{\gamma}\right) = \beta \bar c+\beta^2>0.
\end{equation}
 On the other hand, $c_0$ is a subsolution to \eqref{HJB_u_2'}: for any $m>0$  
\begin{align}
\mathcal{L}&c_0(m) = am\left(\sup_{h\ge 0}\left\{g(h) - \frac{1-\gamma}{\gamma} \frac{\bar c}{a m} h\right\}-\beta\right)\nonumber\\ 
&= am \bigg[g\left(I\left(\frac{1-\gamma}{\gamma} \left[1+\frac{\bar c}{am}\right]\right)\right)
- \frac{1-\gamma}{\gamma} \left[1+\frac{\bar c}{am}\right] I\left(\frac{1-\gamma}{\gamma} \left[1+\frac{\bar c}{am}\right]\right)-\beta\bigg] <0,\label{c0 subsolution}
\end{align}
where $a:= \frac{1-\zeta^{1-\gamma}}{\gamma}$ and the inequality follows from \eqref{g<beta}. Thus, $\Pi$ contains at least $(c_0,c_0 +\beta)$. Also note that for each $(p,q)\in\Pi$, $\mathcal{S}(p,q)\neq\emptyset$ because by construction $q\in \mathcal{S}(p,q)$.
\end{remark}

We first present a basic result for strictly increasing, concave functions $f$ which are bounded by $c_0$ and $c_0+\beta$. Note that the concavity of $f$ implies $f'(\infty) := \lim_{m\to\infty} f'(m-)$ is well-defined. 

\begin{lemma}\label{lem:u(m)/m decrease}
Assume $0<\gamma<1$ and \eqref{bar c}. For any nonnegative, strictly increasing, and concave $f:\R_+\to\R$, we have $\frac{f(m)}{m f'(m-)}\ge1$ for all $m\in(0,\infty)$. If $f$ additionally satisfies $c_0\le f\le c_0+\beta$ on $(0,\infty)$, then $f'(\infty)  =\frac{1-\zeta^{1-\gamma}}{\gamma}$ and $\frac{f(m)}{mf'(m-)}\to 1$ as $m\to\infty$.
\end{lemma}

\begin{proof}
Since $f$ is strictly increasing with $f(0)\ge 0$, the concavity of $f$ implies $\frac{f(m)}{m}\ge f'(m-) >0$, and thus $\frac{f(m)}{m f'(m-)}\ge 1$ for all $m\in(0,\infty)$. Suppose $f$ is additional bounded by $c_0$ and $c_0+\beta$. Since $c_0$ is a linear function with slope $\frac{1-\zeta^{1-\gamma}}{\gamma}$, if $f'(\infty) \neq \frac{1-\zeta^{1-\gamma}}{\gamma}$, then $f(m)\notin [c_0(m), c_0(m)+\beta]$ for $m$ large enough, a contradiction. Moreover, we deduce from $c_0(m)\le f(m)\le c_0(m)+\beta$ and $f'(\infty) = \frac{1-\zeta^{1-\gamma}}{\gamma}>0$ that 
\[
1+ \frac{\delta+(\gamma-1)r}{(1-\zeta^{1-\gamma}) m}\le \frac{f(m)}{m f'(\infty)} \le 1+\frac{\delta+(\gamma-1)r+\beta\gamma}{(1-\zeta^{1-\gamma}) m},\quad m>0.
\] 
This implies $\frac{f(m)}{mf'(\infty)}\to 1$ as $m\to\infty$.
\end{proof}

The next result shows that $c_0+\alpha$ is a supersolution to \eqref{HJB_u_2'} on $(0,\infty)$ for $\alpha$ large enough.

\begin{lemma}\label{lem:c0+alpha}
Assume $0<\gamma<1$, \eqref{bar c}, and \eqref{g<beta}. For any $\alpha\in [0,\beta]$, $c_0+\alpha$ is a supersolution to \eqref{HJB_u_2'} on $(0,\infty)$ if and only if $\alpha\in[\beta_g, \beta]$, where $\beta_g$ is defined in \eqref{betag}.
Specifically,
\begin{align*}
&\alpha\in[\beta_g, \beta]\implies \mathcal{L}({c_0+\alpha})(m)>0\ \hbox{for all}\ m>0;\\ 
&\alpha\in[0,\beta_g)\implies \mathcal{L}({c_0+\alpha})(m)\to-\infty\ \hbox{as}\ m\to\infty. 
\end{align*}
\end{lemma}

\begin{proof}
For any $a,b>0$, consider the function
\begin{equation}\label{theta}
\theta(m) := \mathcal{L}(am+b) = \left(am+ b\right)\left[\left(a-\frac{1-\zeta^{1-\gamma}}{\gamma}\right) m + (b-\bar c)\right] + am \left(\ell(m)-\beta\right),
\end{equation}
where 
\begin{align}\label{ell}
\ell(m) &:= \sup_{h\ge 0}\left\{g(h)-\frac{1-\gamma}{\gamma}\left[1+\frac{b}{am} \right]h\right\}\nonumber\\
&= g\left(I\left(\frac{1-\gamma}{\gamma}\left[1+\frac{b}{am} \right]\right)\right)-\frac{1-\gamma}{\gamma}\left[1+\frac{b}{am} \right] I\left(\frac{1-\gamma}{\gamma}\left[1+\frac{b}{am} \right]\right).
\end{align} 
By direct calculation, 
\begin{align}
\ell'(m) &= \frac{1-\gamma}{\gamma}\frac{b}{a m^2} I\left(\frac{1-\gamma}{\gamma}\left[1+\frac{b}{am} \right]\right),\label{ell'}\\
\theta'(m) &= 2a\left(a-\frac{1-\zeta^{1-\gamma}}{\gamma}\right)m + a(b-\bar c) + b\left(a-\frac{1-\zeta^{1-\gamma}}{\gamma}\right)\nonumber\\
&\ \ \ \ \ + a \left(\ell(m)-\beta + \frac{1-\gamma}{\gamma}\frac{b}{am} I\left(\frac{1-\gamma}{\gamma}\left[1+\frac{b}{a m}\right] \right)\right)\label{theta'},\\
\theta''(m) &=  2a\left(a-\frac{1-\zeta^{1-\gamma}}{\gamma}\right)  - \left(\frac{1-\gamma}{\gamma}\right)^2\frac{b^2}{a m^3} I'\left(\frac{1-\gamma}{\gamma}\left[1+\frac{b}{a m}\right] \right)>0, \label{theta''}
\end{align}
where the positivity follows from $I'(y) = \frac{d}{dy}(g')^{-1}(y) = \frac{1}{g''((g')^{-1}(y))}= \frac{1}{g''(I(y))}$ and $g''<0$.

For any $\alpha\in[0,\beta]$, $c_0(m)+\alpha = am+b$ with $a=\frac{1-\zeta^{1-\gamma}}{\gamma}$ and $b=\bar c+\alpha$. Then \eqref{theta} reduces to 
\begin{equation}\label{theta a=max}
\theta(m) =  am [\alpha-(\beta-\ell(m))] + \alpha b. 
\end{equation}
Observe from \eqref{ell} that $\beta-\ell(m)\to\beta_g$ as $m\to\infty$. It follows that 
\begin{equation}\label{theta at infinity}
\lim_{m\to\infty} \mathcal{L}(c_0+\alpha)(m) = \lim_{m\to\infty}\theta(m)  = \begin{cases}
-\infty,\quad &\hbox{if}\ \alpha\in[0,\beta_g);\\
+\infty,\quad &\hbox{if}\ \alpha\in(\beta_g,\beta].
\end{cases}
\end{equation}
This already shows that  if $\alpha\in[0,\beta_g)$, $c_0+\alpha$ cannot be a supersolution to \eqref{HJB_u_2'} on $(0,\infty)$.

It remains to show that if $\alpha\in[\beta_g,\beta]$, $c_0+\alpha$ satisfies $\mathcal{L}({c_0+\alpha})(m)>0$ for all $m>0$. This is true for $\alpha = \beta$, as explained in Remark~\ref{rem:not empty}. For any $\alpha\in(\beta_g,\beta)$, using $a=\frac{1-\zeta^{1-\gamma}}{\gamma}$ and $b=\bar c+\alpha$ under current setting and \eqref{ell}, \eqref{theta'} becomes
\begin{align}\label{theta' a=max}
\theta'(m) &=  a\left[\alpha -\beta+g\left(I\left(\frac{1-\gamma}{\gamma}\left[1+\frac{b}{a m} \right]\right)\right)-\frac{1-\gamma}{\gamma} I\left(\frac{1-\gamma}{\gamma}\left[1+\frac{b}{a m} \right]\right)\right],
\end{align}
which implies $\theta'(m)\to \alpha -\beta <0$ as $m\downarrow 0$. This, together with $\lim_{m\to\infty}\theta(m)=\infty$ in \eqref{theta at infinity} and $\theta''(\cdot)>0$ in \eqref{theta''}, shows that $\theta$ must attain a global minimum at some $m^*\in(0,\infty)$. Using $\theta'(m^*)=0$, we obtain from \eqref{theta' a=max} that
\begin{equation}\label{at m^*}
\alpha - \frac{1-\gamma}{\gamma}  I\left(\frac{1-\gamma}{\gamma}\left[1+\frac{b}{a m^*} \right]\right) = \beta-g\left(I\left(\frac{1-\gamma}{\gamma}\left[1+\frac{b}{am^*} \right]\right)\right).
\end{equation}
The global minimum can then be computed as 
\begin{align*}
\theta(m^*) &= am^*[\alpha-\beta + \ell(m)]+ \alpha b\\
&= am^*\bigg[\alpha - \beta+g\left(I\left(\frac{1-\gamma}{\gamma}\left[1+\frac{b}{am^*} \right]\right)\right)\\
&\hspace{0.7in}-\frac{1-\gamma}{\gamma} \left[1+\frac{b}{am^*} \right]I\left(\frac{1-\gamma}{\gamma}\left[1+\frac{b}{am^*} \right]\right)\bigg] +\alpha b\\
&= b\left[\alpha-\frac{1-\gamma}{\gamma}  I\left(\frac{1-\gamma}{\gamma}\left[1+\frac{b}{am^*} \right]\right) \right]\\
&= b\left[\beta-g\left(I\left(\frac{1-\gamma}{\gamma}\left[1+\frac{b}{am^*} \right]\right)\right)\right] >0,
\end{align*}
 where the second equality comes from \eqref{ell}, the third and fourth equalities follow from \eqref{at m^*}, and the final inequality is due to \eqref{g<beta}. We thus conclude that for any $\alpha\in(\beta_g,\beta)$, $\mathcal{L}(c_0+\alpha)(m) = \theta(m) \ge \theta(m^*) >0$ for all $m\in (0,\infty)$. Finally, for $\alpha=\beta_g$, since $c_0+\beta_g$ is the pointwise infimum of the supersolutions $c_0+\alpha$, $\alpha\in(\beta_g,\beta]$, it must also be a supersolution, thanks to \cite[Lemma 4.2]{CrandallIshiiLions1992}. Observe from \eqref{theta' a=max} that $\lim_{m\uparrow\infty}\theta'(m) = a[\beta_g-\beta_g] =0$. This, together with $\theta''>0$ in \eqref{theta''} and $c_0+\beta_g$ being a supersolution, shows that $\theta(m)$ must be strictly decreasing on $(0,\infty)$ with $\lim_{m\uparrow\infty}\theta(m)\ge 0$. This already implies $\mathcal{L}(c_0+\beta_g)(m)=\theta(m)>0$ for all $m>0$.   
\end{proof}


\begin{lemma}\label{lem:p<c0+betag}
Assume $0<\gamma<1$, \eqref{bar c}, and \eqref{g<beta}. For any $(p,q)\in\Pi$, $p < c_0+\beta_g$ on $\R_+$, with $\beta_g$ defined in \eqref{betag}.
\end{lemma}

\begin{proof}
By contradiction, suppose ``$p < c_0+\beta_g$ on $\R_+$'' does not hold. If there exists $m_0>0$ such that $p(m)= c_0(m)+\beta_g$ for all $m\ge m_0$, then the subsolution property of $p$ is violated, thanks to Lemma~\ref{lem:c0+alpha}. Thus, it remains to deal with the second case: there exists $m_0>0$ such that $p(m)> c_0(m)+\beta_g$ for $m>m_0$. 

Since $p'(\infty)=\frac{1-\zeta^{1-\gamma}}{\gamma}$ (by Lemma~\ref{lem:u(m)/m decrease} (i)), there must exist $\alpha_0\in (\beta_g,\beta]$ such that $p\le c_0+\alpha_0$  and $(c_0(m)+\alpha)-p_0(m)\downarrow 0$ as $m\to\infty$. Consider the collection of functions $\{c_0+\alpha:\alpha\in (\frac{\alpha_0+\beta_g}{2},\alpha_0)\}$. For each $\alpha\in (\frac{\alpha_0+\beta_g}{2},\alpha_0)$, we let $\theta^\alpha(m):= \mathcal{L}(c_0+\alpha)(m)$, and recall the formula of $\theta^\alpha$ in \eqref{theta a=max}. It shows that 
\[
\frac{\theta^\alpha(m)}{m} = a[\alpha-(\beta-\ell(m))]+\frac{\alpha (\bar c+\alpha)}{m} \to a[\alpha-\beta_g]\quad \hbox{as}\ m\to\infty,
\]
where $a:=\frac{1-\zeta^{1-\gamma}}{\gamma}$. Moreover, in view of \eqref{ell} with $b$ replaced by $\bar c+\alpha$, the above convergence is uniform in $\alpha\in(\frac{\alpha_0+\beta_g}{2},\alpha_0)$. That is, for any $\delta>0$, there exists $M(\delta)>0$ such that for all $m\ge M(\eps)$ and $\alpha\in(\frac{\alpha_0+\beta_g}{2},\alpha_0)$, $|\frac{\theta^\alpha(m)}{m}- a[\alpha-\beta_g]|<\delta$. Taking $\delta:= \frac{a(\alpha_0-\beta_g)}{4}$, we get for any $m>M(\delta)$ and $\alpha\in(\frac{\alpha_0+\beta_g}{2},\alpha_0)$,
\begin{equation}\label{uniform lower bdd}
\frac{\theta^\alpha(m)}{m} > a[\alpha-\beta_g] - \delta > a\left[\frac{\alpha_0+\beta_g}{2}-\beta_g\right] -\delta = a\left[\frac{\alpha_0-\beta_g}{4}\right].
\end{equation}
Fix $\eps\in(0,\frac{a(\alpha_0-\beta_g)}{4\beta})$. We can take $\alpha\in(\frac{\alpha_0+\beta_g}{2},\alpha_0)$ large enough such that $c_0+\alpha$ intersects $p$ at $m^*> M(\delta)$ and $a<p'(m^*)< a+\eps$. It follows that 
\begin{align*}
\mathcal{L}(p)(m^*) &= p(m^*)(p(m^*)-c_0(m^*)) + m^* p'(m^*) \left(\sup_{h\ge 0}\left\{g(h)-\frac{1-\gamma}{\gamma}\frac{p(m^*)}{m^*p'(m^*)}h\right\}-\beta\right)\\
&= \alpha (c_0(m^*)+\alpha) + \sup_{h\ge 0}\left\{g(h)m^*p'(m^*)-\frac{1-\gamma}{\gamma}p(m^*)h\right\}-\beta m^*p'(m^*)\\
&\ge \alpha (c_0(m^*)+\alpha) + \sup_{h\ge 0}\left\{g(h)m^* a -\frac{1-\gamma}{\gamma}(c_0(m^*)+\alpha)h\right\}-\beta m^*(a+\eps)\\
&= \theta^\alpha(m^*) - \beta m^*\eps = m^*\left(\frac{\theta^\alpha(m^*)}{m^*}-\beta\eps\right) >0,
\end{align*}
where the third equality follows from $\theta^\alpha(m^*)= \mathcal{L}(c_0+\alpha)(m^*)$, and the last inequality is due to \eqref{uniform lower bdd} and the choice of $\eps$. This implies that $p$ cannot be a subsolution to \eqref{HJB_u_2'} on $(0,\infty)$, a contradiction. 
\end{proof}

Following Perron's method, we introduce, for each $(p,q)\in\Pi$, the function
\begin{equation}
u^*_{p,q}(m) := \inf_{f\in\mathcal{S}(p,q)}f(m),\quad m\ge 0.
\end{equation}

\begin{proposition}[Supersolution Property]\label{prop:supersolution}
Assume $0<\gamma<1$ and \eqref{bar c}. For any $(p,q)\in\Pi$, $u^*_{p,q}\in\mathcal{S}(p,q)$.
\end{proposition}

\begin{proof}
As a pointwise infimum of concave nondecreasing functions bounded by $c_0$ and $c_0+\beta$, $u^*_{p,q}$ is by definition concave, nondecreasing, and bounded by $c_0$ and $c_0+\beta$. The concavity of $u^*_{p,q}$ yields the desired continuity. Then, by \cite[Lemma 4.2]{CrandallIshiiLions1992}, $u^*_{p,q}$, being continuous and a pointwise infimum of viscosity supersolution, is again a viscosity supersolution. It remains to show that $u^*_{p,q}$ is strictly increasing. Suppose to the contrary that $u^*_{p,q}\equiv \kappa>0$ in a neighborhood of some $m^*\in (0,\infty)$. The concavity of $u^*_{p,q}$ then implies that $u^*_{p,q}\equiv\kappa$ on $[m^*,\infty)$. It follows that $\mathcal{L}u^*_{p,q}(m) = \kappa(\kappa-c_0(m))<0$ as $m$ large enough. This contradicts the supersolution property of $u^*_{p,q}$.  
\end{proof}

\begin{proposition}[Subsolution Property]\label{prop:subsolution}
Assume $0<\gamma<1$ and \eqref{bar c}. Fix $(p,q)\in\Pi$. Suppose $u^*_{p,q}$ is strictly concave at $m_0\in (0,\infty)$ in the following sense: 
\begin{equation}\label{strict concave at a point}
\begin{split}
&\hbox{for any $m_1,m_2\in (0,\infty)$ and $\lambda\in(0,1)$ such that $m_0 = \lambda m_1 + (1-\lambda) m_2$},\\ 
&\hspace{1.5in}\hbox{$u^*_{p,q}(m_0) > \lambda u^*_{p,q}(m_1) + (1-\lambda) u^*_{p,q}(m_2)$.}
\end{split}
\end{equation}
Then, $u^*_{p,q}$ is a viscosity subsolution to \eqref{HJB_u_2'} at $m_0$.
\end{proposition}

\begin{proof}
If $u^*_{p,q}$ is strictly concave at $m_0\in (0,\infty)$ as defined above, there are three possibilities: (i) $(u^*_{p,q})'(m_0 -)\neq (u^*_{p,q})'(m_0+)$; (ii) $(u^*_{p,q})'(m_0 -) = (u^*_{p,q})'(m_0+)$, and $u^*$ is strictly concave on the interval $[m_0-\kappa,m_0+\kappa]$ for some $\kappa>0$; (iii) $(u^*_{p,q})'(m_0 -) = (u^*_{p,q})'(m_0+)$, and there exists $\kappa_1,\kappa_2>0$ such that $u^*_{p,q}$ is linear on $[m_0-\kappa_1,m_0]$ and strictly concave on $[m_0,m_0+\kappa_2]$, or strictly concave on $[m_0-\kappa_1,m_0]$ and linear on $[m_0,m_0+\kappa_2]$. 

We assume, by contradiction, that there exists a test function $\psi\in C^1((0,\infty))$ such that $0 = (u^*_{p,q}-\psi)(m_0) > (u^*_{p,q}-\psi)(m)$ for all $m\in (0,\infty)\setminus\{m_0\}$ and $\mathcal{L}\psi(m_0)>0$. For the cases (i) and (ii), we can assume without loss of generality that $\psi$ is strictly increasing and concave on $(0,\infty)$. Take $\delta>0$ small enough such that $\mathcal{L}\psi(m)>0$ for all $m\in (m_0-\delta,m_0+\delta)$. Then, for small enough $\eps>0$, one can take $0<\delta_1\le \delta$ such that for each $0<\eta\le\eps$, $\mathcal{L}(\psi-\eta)(m)>0$ for all $m\in (m_0-\delta_1,m_0+\delta_1)$. Consider the function
\begin{equation}\label{u^eta}
u^\eta (m) := 
\begin{cases}
\min\{u^*_{p,q}(m),\psi(m)-\eta\},\ &\hbox{for}\ m\in[m_0-\delta_1, m_0+\delta_1],\\
u^*_{p,q}(m),\ &\hbox{for}\ m\notin[m_0-\delta_1, m_0+\delta_1].
\end{cases}
\end{equation}
When $\eta$ is small enough, $u^\eta$ by construction is a concave, strictly increasing viscosity supersolution to \eqref{HJB_u_2'} on $(0,\infty)$, and $u^*_{p,q}-\eta\le u^\eta\le u^*_{p,q}$. That is, $u^\eta\in \mathcal{S}(p,q)$ as $\eta$ is small enough. However, by definition $u^\eta< u^*_{p,q}$ in some small neighborhood of $m_0$, which contradicts the definition of $u^*_{p,q}$. 

Now we deal with the case (iii). Set $a:= (u^*_{p,q})'(m_0 -) = (u^*_{p,q})'(m_0+)$. In view of \eqref{HJB_u_2'}, to get the desired subsolution property, it suffices to prove
\begin{equation}\label{for subsolution}
(u^*_{p,q})^2(m_0) - c_0(m_0)u^*_{p,q}(m_0)+ a m_0  \left(\sup\limits_{h\ge 0}\left\{g(h) - \frac{1-\gamma}{\gamma} \frac{u^*_{p,q}(m_0)}{a m_0} h\right\}-\beta\right)\le 0.
\end{equation}
We assume, without loss of generality, that $u^*_{p,q}$ is linear on $[m_0-\kappa_1,m_0]$ and strictly concave on $[m_0,m_0+\kappa_2]$. Take $\{\ell_n\}_{n\in\N}$ in $(m_0,m_0+\kappa_2]$ such that $\ell_n\downarrow m_0$ and $u^*_{p,q}$ is differentiable at $\ell_n$. Then, the subsolution property we established above under case (ii) implies that $\mathcal{L} u^*_{p,q}(\ell_n) \le 0$ for all $n\in\N$. 
Observe that the map 
\begin{equation}\label{sup conti.}
m\mapsto \sup\limits_{h\ge 0}\left\{g(h) - \frac{1-\gamma}{\gamma} \frac{u^*_{p,q}(m)}{m(u^*_{p,q})'(m)} h\right\}\quad \hbox{is continuous around $m_0$}, 
\end{equation}
thanks to $g$ being strictly concave and nondecreasing with $g'(\infty)=0$. As $n\to\infty$, $\mathcal{L} u^*_{p,q}(\ell_n) \le 0$ implies \eqref{for subsolution}, by the continuity of $u^*_{p,q}$ and \eqref{sup conti.}. 
\end{proof}

We next establish the strict concavity of $u^*_{p,q}$. Recall that $I$ denotes the inverse function of $g'$.

\begin{proposition}[Strict Concavity]\label{prop:strict concave}
Assume $0<\gamma<1$, \eqref{bar c}, and \eqref{g<beta}. For any $(p,q)\in\Pi$, $u^*_{p,q}$ is strictly concave on $(0,\infty)$.
\end{proposition}

\begin{proof}
Assume, by contradiction, that $u^*_{p,q}$ is linear, i.e. $u^*_{p,q}(m) =am+b$, on some interval of $\R_+$. Since $u^*_{p,q}\in\mathcal{S}(p,q)$, we deduce from Lemma~\ref{lem:u(m)/m decrease} that $a\ge \frac{1-\zeta^{1-\gamma}}{\gamma}$ and $b\in[\bar c, \bar c+\beta]$. Recall $\theta(m):=\mathcal{L}(am+b)$ in \eqref{theta}. 
\begin{itemize}[leftmargin=*]
\item{\bf Case I:} $a=\frac{1-\zeta^{1-\gamma}}{\gamma}$ and $b\in[\bar c,\bar c+\beta_g)$.
Then $u^*_{p,q}(m) = am+b = c_0(m)+\alpha$ for $m$ large enough, where $\alpha:= b-\bar c\in[0,\beta_g)$. By Lemma~\ref{lem:c0+alpha}, $\lim_{m\uparrow\infty}\mathcal{L}(u^*_{p,q})(m) = \lim_{m\uparrow\infty}\mathcal{L}(c_0+\alpha)(m)=-\infty$. This contradicts the supermartingale property of $u^*_{p,q}$.

\item{\bf Case II:} $a=\frac{1-\zeta^{1-\gamma}}{\gamma}$ and $b\in[\bar c+\beta_g, \bar c+\beta]$.
		\begin{itemize}[leftmargin=*]
		\item {\bf Case II-1:} $u^*_{p,q}(m) = am+b$ for all $m\ge 0$, with $b\in(\bar c+\beta_g, \bar c+\beta]$.\\  
		Let us write $u^*_{p,q}(m) = c_0(m)+\alpha$, with $\alpha:= b-\bar c\in (\beta_g,\beta]$. For any $\bar \alpha\in (\beta_g,\alpha)$, Lemmas~\ref{lem:c0+alpha} and \ref{lem:p<c0+betag} imply that $c_0+\bar \alpha$ belongs to $\mathcal{S}({p,q})$ and is strictly less than $u^*_{p,q}$, which contradicts the definition of $u^*_{p,q}$. 
		
		\item {\bf Case II-2:} $u^*_{p,q}(m) = am+(\bar c +\beta_g)$ for all $m\ge 0$.\\
We deduce from \eqref{theta} and \eqref{theta' a=max} that $\lim_{m\downarrow 0}\theta(m)= b(b-\bar c)>0$ and $\lim_{m\downarrow 0}\theta'(m)= b-\bar c-\beta <0$. Thus, we can take $m^*>0$ small enough such that $\theta(m^*)>0$ and $\theta'(m^*)<0$. In view of the continuous dependence of $\theta(m^*)$ and $\theta'(m^*)$ on $a,b$ in \eqref{theta} and \eqref{theta'}, there exists $\delta> 0$ small enough such that when $a,b$ are replaced by $\bar a \in (a,a+\delta)$ and $\bar b \in (b-\delta,b)$, $\theta(m^*)>0$ and $\theta'(m^*)<0$ still hold. Take suitable $\bar a \in (a,a+\delta)$ and $\bar b \in (b-\delta,b)$ such that $\bar a m^*+\bar b = u^*_{p,q}(m^*)$ and $\bar a m +\bar b > p(m)$ for $m\in(0,m^*]$ (this is doable thanks to Lemma~\ref{lem:p<c0+betag}). For clarity, let $\bar \theta$ and $\bar \theta'$ denote $\theta$ and $\theta'$ with $a,b$ replaced by $\bar a,\bar b$. Now, we deduce from $\lim_{m\downarrow 0}\bar \theta(m)=\bar b(\bar b-\bar c)>0$ (obtained from \eqref{theta} as above), $\bar\theta(m^*)>0$, $\bar\theta'(m^*)<0$, and $\bar \theta''(m)>0$ for all $m> 0$ (by \eqref{theta''}) that $\bar\theta(m)>0$ for all $m\in(0, m^*)$. Consider the function $\psi(m):= \bar a m+\bar b$. By definition $\mathcal{L}\psi(m)=\bar\theta(m)>0$ for $m\in(0,m^*)$. Thus, $\psi\wedge u^*_{p,q}$ belongs to $\mathcal{S}(p,q)$ and is strictly less than $u^*_{p,q}$ for $m\in (0, m^*)$. This contradicts the definition of $u^*_{p,q}$. 

		\item {\bf Case II-3:} There exists $m_0>0$ such that $u^*_{p,q}(m) = am+b$ for all $m\ge m_0$, and $u^*_{p,q}$ is strictly concave at $m_0$ in the sense of \eqref{strict concave at a point}.\\
 By Proposition~\ref{prop:subsolution}, $u^*_{p,q}$ is a viscosity subsolution to \eqref{HJB_u_2'} at $m_0$. Take $\psi(m) := am+b$, $m\in (0,\infty)$, as a test function of $u^*_{p,q}$ at $m_0$. The subsolution property of $u^*_{p,q}$ yields $\mathcal{L}\psi(m_1)\le 0$. Note that $\psi(m) = c_0(m)+\alpha$ with $\alpha:=b-\bar c\in[\beta_g,\beta]$. Thus, by Lemma~\ref{lem:c0+alpha}, $\mathcal{L}{\psi}(m)>0$ for all $m>0$, a contradiction.
		\end{itemize}
		
\item{\bf Case III:} $a>\frac{1-\zeta^{1-\gamma}}{\gamma}$ and $b=\bar c$. 
Then there exists $m_0>0$ such that $u^*_{p,q}(m) = am+b$ for $m\in[0,m_0]$ and $u^*_{p,q}$ is strictly concave at $m_0$ in the sense of \eqref{strict concave at a point}. By Proposition~\ref{prop:subsolution}, $u^*_{p,q}$ is a viscosity subsolution to \eqref{HJB_u_2'} at $m_0$. 
For all $m\in (0,\infty)$, define 
\begin{align*}
\eta(m)&:= \left(a+\frac bm\right)\left[\left(a-\frac{1-\zeta^{1-\gamma}}{\gamma}\right) m + (b-\bar c)\right] + a \left(\ell(m)-\beta\right),
\end{align*}
with $\ell$ as in \eqref{ell}. Note that $\theta(m) = m\eta(m)$. By direct calculation and \eqref{ell'}, 
\begin{align}
\eta'(m) &= a\left(a-\frac{1-\zeta^{1-\gamma}}{\gamma}\right) - \frac{b}{m^2}\left[(b-\bar c) - \frac{1-\gamma}{\gamma} I\left(\frac{1-\gamma}{\gamma}\left[1+\frac{b}{a m} \right]\right)\right].\label{eta'}
\end{align}
Since we currently have $b = \bar c$, $\eta'(m) >0$ for all $m\in(0,\infty)$. Now, take $\psi(m) := am+b$, $m\in (0,\infty)$, as a test function of $u^*_{p,q}$ at $m_0$. Then the subsolution property of $u^*_{p,q}$ implies $0\ge \mathcal{L}\psi(m_0)=\theta(m_0)=m_0\eta(m_0)$. We therefore have $\eta(m)<0$ for all $m\in(0,m_0)$. The supersolution property of $u^*_{p,q}$, however, entails $0\le \mathcal{L}u^*_{p,q}(m)=\theta(m)=m\eta(m)$ for all $m\in(0,m_0)$, a contradiction.

\item{\bf Case IV:} $a> \frac{1-\zeta^{1-\gamma}}{\gamma}$ and $b\in(\bar c, \bar c+\beta)$.
		\begin{itemize}[leftmargin=*]
		\item {\bf Case IV-1:} There exists $m_0>0$ such that $u^*_{p,q}(m) = am+b$ for $m\in[0,m_0]$, and $u^*_{p,q}$ is strictly concave at $m_0$ in the sense of \eqref{strict concave at a point}.\\
		We first show that $p(0)$ has to be strictly less than $u^*_{p,q}(0)$. If $p(0)= u^*_{p,q}(0)$, then $\lim_{m\downarrow 0}p'(m) \le a$; otherwise, $p(m)>u^*_{p,q}(m)$ for $m> 0$ small enough, which contradicts $u^*_{p,q}\in \mathcal{S}(p,q)$. By the concavity of $p$, we can take a real sequence $\{\ell_n\}$ such that $\ell_n\downarrow 0$ and $p$ is differentiable at $\ell_n$. The subsolution property of $p$ then implies $\mathcal{L}p(\ell_n)\le 0$ for all $n\in\N$. As $n\to\infty$, we get $p(0)(p(0)-\bar c)\le 0$, thanks to the finiteness of $\lim_{m\downarrow 0}p'(m)$. This shows that $p(0)<\bar c$, a contradiction to $p\ge c_0$.
		
		By Proposition~\ref{prop:subsolution}, $u^*_{p,q}$ is a viscosity subsolution to \eqref{HJB_u_2'} at $m_0$. Take $\psi(m) := am+b$, $m\in (0,\infty)$, as a test function of $u^*_{p,q}$ at $m_0$. Then the subsolution property of $u^*_{p,q}$ implies $0\ge \mathcal{L}\psi(m_0)=\theta(m_0)$. Observe from \eqref{theta} that $\lim_{m\downarrow 0}\theta(m)= b(b-\bar c)>0$. If $\lim_{m\downarrow 0}\theta'(m)\ge 0$, then $\theta''>0$ on $(0,\infty)$ (by \eqref{theta''}) implies that $\theta(m)>\theta(0)>0$ for all $m>0$, which contradicts $\theta(m_0)\le 0$. If $\lim_{m\downarrow 0}\theta'(m)< 0$, then we can follow the argument in Case II-2. Take $0<m^*<m_0$ small enough such that $\theta(m^*)>0$ and $\theta'(m^*)<0$. By the continuous dependence of $\theta(m^*)$ and $\theta'(m^*)$ on $a,b$, there exists $\delta> 0$ such that when $a,b$ are replaced by $\bar a \in (a,a+\delta)$ and $\bar b \in (b-\delta,b)$, $\theta(m^*)>0$ and $\theta'(m^*)<0$ still hold. Choose suitable $\bar a \in (a,a+\delta)$ and $\bar b \in (b-\delta,b)$ such that $\bar a m^*+\bar b = u^*_{p,q}(m^*)$ and $\bar a m +\bar b > p(m)$ for $m\in(0,m^*]$ (this is doable thanks to $p(0)<u^*_{p,q}(0)$). For clarity, let $\bar \theta$ and $\bar \theta'$ denote $\theta$ and $\theta'$ with $a,b$ replaced by $\bar a,\bar b$. Now, we deduce from $\lim_{m\downarrow 0}\bar \theta(m)>0$, $\bar\theta(m^*)>0$, $\bar\theta'(m^*)<0$, and $\bar \theta''>0$ on $(0,\infty)$ that $\bar\theta(m)>0$ for all $m\in(0, m^*)$. Consider the function $\phi(m):= \bar a m+\bar b$. By definition $\mathcal{L}\phi(m)=\bar\theta(m)>0$ for $m\in(0,m^*)$. Thus, $\phi\wedge u^*_{p,q}$ belongs to $\mathcal{S}(p,q)$ and is strictly less than $u^*_{p,q}$ for $m\in (0, m^*)$. This contradicts the definition of $u^*_{p,q}$. 
		\item {\bf Case IV-2:} There exist $m_1, m_2\in (0,\infty)$ such that $u^*_{p,q}(m) = am+b$ for $m\in[m_1,m_2]$, and $u^*_{p,q}$ is strictly concave at $m_1$ and $m_2$ in the sense of \eqref{strict concave at a point}.\\
		By Proposition~\ref{prop:subsolution}, $u^*_{p,q}$ is a viscosity subsolution to \eqref{HJB_u_2'} at both $m_1$ and $m_2$. Now, take $\psi(m) := am+b$, $m\in (0,\infty)$, as a test function of $u^*_{p,q}$ at $m_1$ and $m_2$. Then the subsolution property of $u^*_{p,q}$ implies $0\ge \mathcal{L}\psi(m_1)=\theta(m_1)$ and $0\ge \mathcal{L}\psi(m_2)=\theta(m_2)$. Since $\theta''>0$ on $(0,\infty)$ (by \eqref{theta''}), we must have $\theta(m_3)<0$ for some $m_3\in(m_1,m_2)$. The supersolution property of $u^*_{p,q}$, however, entails $0\le \mathcal{L}\psi(m)=\theta(m)$ for all $m\in(m_1,m_2)$, a contradiction.
		\end{itemize}
\end{itemize}
\end{proof}

\begin{proposition}[Regularity]\label{prop:regularity}
Assume $0<\gamma<1$, \eqref{bar c}, and \eqref{g<beta}. For any $(p,q)\in\Pi$, $u^*_{p,q}$ is a strictly concave classical solution to \eqref{HJB_u_2'} on $(0,\infty)$.
\end{proposition}

\begin{proof}
For any  $(p,q)\in\Pi$, Propositions~\ref{prop:supersolution}, \ref{prop:subsolution}, and \ref{prop:strict concave} immediately imply that $u^*_{p,q}$ is a strictly concave viscosity solution to \eqref{HJB_u_2'} on $(0,\infty)$. It remains to show that $u^*_{p,q}$ is differentiable everywhere on $(0,\infty)$. Assume, by contradiction, that there exists $m_0\in (0,\infty)$ such that $a:=(u^*)'(m_0+) < (u^*)'(m_0-)=:b$. Take $\{k_n\}_{n\in\N}$ and $\{\ell_n\}_{n\in\N}$ in $(0,\infty)$ such that $k_n\uparrow m_0$, $\ell_n\downarrow m_0$, and $u^*_{p,q}$ is differentiable at $k_n$ and $\ell_n$ for all $n\in\N$. By the viscosity solution property of $u^*_{p,q}$, $\mathcal{L}u^*(k_n) = \mathcal{L}u^*(\ell_n)=0$, for all $n\in\N$. As $n\to \infty$, we get 
\begin{equation*}
\sup\limits_{h\ge 0}\left\{(g(h)-\beta)a - \frac{1-\gamma}{\gamma} \frac{u^*_{p,q}(m_0)}{m_0 } h\right\} = \sup\limits_{h\ge 0}\left\{(g(h)-\beta)b - \frac{1-\gamma}{\gamma} \frac{u^*_{p,q}(m_0)}{m_0 } h\right\},
\end{equation*}
which implies that $a=b$, a contradiction.
\end{proof}

\begin{proposition}[Verification]\label{prop:uniqueness}
Assume $0 < \gamma < 1$, \eqref{bar c}, and \eqref{g<beta}. If $u:\R_+\to\R_+$ is a nonnegative, strictly increasing, and concave classical solution to \eqref{HJB_u_2'} on $(0,\infty)$, then
\[
V(x,m)=\frac{x^{1-\gamma}}{1-\gamma} u(m)^{-\gamma}\quad \hbox{for all}\ (x,m)\in\R^2_+. 
\]
Furthermore, $(\hat c,\hat h)$ defined by
\[
\hat c_t := u(M_t)\quad \hbox{and} \quad \hat h_t := I\left(\frac{1-\gamma}{\gamma}\frac{u(M_t)}{M_t\cdot (u)'(M_t)}\right),\quad \hbox{for all}\ t\ge 0,
\]
is an optimal control of \eqref{problem}. 
\end{proposition}

\begin{proof}
Set $w(x,m) := \frac{x^{1-\gamma}}{1-\gamma} u(m)^{-\gamma}$. In view of \thmref{thm:verification}, it suffices to show that \eqref{t to infty} and \eqref{n to infty} hold, and $(\hat c,\hat h)$ belongs to $\A$. Since $u$ is nonnegative, strictly increasing, and concave, Lemma~\ref{lem:u(m)/m decrease} implies that $\hat h_t \le I(\frac{1-\gamma}{\gamma})$ for all $t\ge 0$. Moreover, there exist $a,b>0$ such that $u(m)<am+b$ for all $m\ge 0$. It follows that for any compact subset $K\subset \R_+$, 
\begin{equation*}
\int_K \hat c_t dt \le \int_K aM_t +b\ dt \le \int_K  a m e^{\beta t} +b\ dt <\infty.  
\end{equation*}
This already shows that $(\hat c,\hat h)\in\A$. 

Under \eqref{g<beta}, $u$ being a classical solution to \eqref{HJB_u_2'} implies $ u^2(m) - u(m)  c_0(m) \ge 0$, and thus $ u(m)\ge  c_0(m)$ for all $m\in(0,\infty)$. Now, for any $(x,m)\in \R^2_+$, $(c,h)\in \A$, and $n\in\N$, by using $0<\gamma<1$ and $X^{0,x,c,h}_t \le X^{0,x,c,h}_{\tau_{n}}\exp\left(r(t-\tau_{n})\right)$, we have
\begin{align*}
0&\le \E\left[\exp\left(-\int_{\tau_{n}}^t (\delta+ M^{0,m,h}_s) ds\right) w\left(\zeta^n X^{0,x,c^{},h}_t, M^{0,m,h}_t\right)\ \middle|\ Z_1,...,Z_n\right]\\
 &\le e^{-\delta(t-\tau_n)} \frac{(X^{0,x,c,h}_{\tau_{n}})^{1-\gamma}}{1-\gamma} e^{(1-\gamma)r(t-\tau_{n})} \E[u(M^{0,m,h}_t)^{-\gamma}\mid Z_1,...,Z_n]\\
& \le e^{-(\delta+(\gamma-1)r)(t-\tau_n)} \frac{(X^{0,x,c,h}_{\tau_{n}})^{1-\gamma}}{1-\gamma} (\bar{c})^{-\gamma}\to 0\ \ \hbox{a.s.}\quad \hbox{as}\ t\to\infty,
\end{align*}
where the second line is due to $M^{0,m,h}_t \ge 0$ by definition, and the third line follows from $u$ being strictly increasing with $u(m)> u(0)\ge c_0(0)=\bar c>0$, and the convergence is a consequence of $\delta + (\gamma-1) r = \gamma \bar c>0$. This in particular implies \eqref{t to infty}. On the other hand, for each $n\in\N$,
\begin{equation*}
\begin{split}
0&\le \E\left[e^{-\delta \tau_{n}} w\left(\zeta^n X^{0,x,c,h}_{\tau_{n}},M^{0,m,h}_{\tau_{n}}\right)\right]\le \zeta^{(1-\gamma)n}\frac{x^{1-\gamma}}{1-\gamma}  \E\left[e^{-\delta \tau_{n}}e^{(1-\gamma) r\tau_{n}} u(M^{0,m,h}_{\tau_{n}})^{-\gamma}\right]\\
&\le \zeta^{(1-\gamma)n}\frac{x^{1-\gamma}}{1-\gamma} (\bar c)^{-\gamma} \E\left[e^{-(\delta+(\gamma-1)r) \tau_{n}}\right]\le \zeta^{(1-\gamma)n}\frac{x^{1-\gamma}}{1-\gamma} (\bar c)^{-\gamma}\to 0 \quad \hbox{as}\ t\to\infty,
\end{split}
\end{equation*}
where the last inequality is due to $\delta+(\gamma-1)r =\gamma \bar c>0$. The shows that \eqref{n to infty} is also satisfied. 
\end{proof}

Proposition \ref{prop:uniqueness}, together with Propositions~\ref{prop:regularity} and \ref{prop:supersolution}, leads to:

\begin{corollary}\label{coro:indep. of p,q}
Assume $0 < \gamma < 1$, \eqref{bar c}, and \eqref{g<beta}. Then $u^*_{p,q}$ is independent of the choice of $(p,q)\in\Pi$, and it is the unique nonnegative, strictly increasing, and concave classical solution to \eqref{HJB_u_2'} on $(0,\infty)$.
\end{corollary}

\begin{remark}
Proposition \ref{prop:uniqueness} and Corollary \ref{coro:indep. of p,q} yield Theorem~\ref{thm:main 3}.
\end{remark}

In the following, we will simply denote by $u^*$ the function $u^*_{p,q}$ for any $(p,q) \in\Pi$. 

\begin{corollary}[Strict Concavity of $u_0$]\label{coro:u0 concave}
Assume $0 < \gamma < 1$ and \eqref{bar c}. Then $u_0$, defined in \eqref{u0}, is strictly concave on $(0,\infty)$.
\end{corollary}

\begin{proof}
With $g\equiv 0$, the equation \eqref{HJB_u_2'} reduces to \eqref{HJB_u_1}, and we can repeat the same arguments in this section (with much simpler proofs) to show that the strictly concave $u^*$ constructed under Perron's method coincides with $u_0$. 
\end{proof}

Now, we are ready to prove Theorem~\ref{thm:main 4}

\begin{proof}[Proof of Theorem~\ref{thm:main 4}]
First, observe that $\beta_g = \beta - g\big(I\big(\frac{1-\gamma}{\gamma}\big)\big) + \frac{1-\gamma}{\gamma} I\big(\frac{1-\gamma}{\gamma}\big)$. Then \eqref{g<beta} implies that $\beta_g>0$.
Since $u_0$ is a solution to \eqref{HJB_u_1} and $u_0'(m)\ge 0$, it is a supersolution to \eqref{HJB_u_2'}. This, together with Lemma~\ref{lem:properties of u0}, Corollary~\ref{coro:u0 concave}, and Remark~\ref{rem:not empty}, shows that $(c_0,u_0)\in\Pi$. It follows that $u^* = u^*_{c_0,u_0}\le u_0$. Similarly, $(c_0,c_0+\beta_g)\in\Pi$ by Lemma~\ref{lem:c0+alpha}, which implies $u^*=u^*_{c_0,c_0+\beta_g}\le c_0+\beta$. This already yields $u^*\le \min\{u_0, c_0+\beta_g\}$. On the other hand, thanks again to Lemma~\ref{lem:properties of u0} and Corollary~\ref{coro:u0 concave} with $\beta$ replaced by $\beta_g$, $u_0^g$ is nonnegative, strictly increasing, concave, and bounded from below and above by $c_0$ and $c_0+
\beta_g$ respectively. Then, Lemma~\ref{lem:u(m)/m decrease} implies $\frac{u^g_0(m)}{m (u^g_0)'(m)}\ge 1$ for all $m>0$. It follows that
\[
\beta -\sup_{h\ge 0}\left\{g(h)-\frac{1-\gamma}{\gamma}\frac{u^g_0(m)}{m (u^g_0)'(m)} h\right\} \ge \beta_g,\quad \forall m>0.
\]  
Since $u_0^g$ is by construction a solution to \eqref{HJB_u_1} with $\beta$ replaced by $\beta_g$, the above inequality gives
\begin{align*}
0 &= (u^g_0(m))^2 - u^g_0(m) c_0(m) - \beta_g m (u^g_0)'(m)\\
&\ge  (u^g_0(m))^2 - u^g_0(m) c_0(m) + m (u^g_0)'(m) \left(\sup_{h\ge 0}\left\{g(h)-\frac{1-\gamma}{\gamma}\frac{u^g_0(m)}{m (u^g_0)'(m)} h\right\} - \beta\right) = \mathcal{L}u^g_0(m),
\end{align*}
for all $m>0$. This shows that $(u^g_0,c_0+\beta_g)\in\Pi$, and thus $u^*=u^*_{u^g_0,c_0+\beta_g}\ge u^g_0$.
\end{proof}



\section{Verification}\label{sec:appendix}
In this section, we provide a general verification theorem for the value function $V(x,m)$ in \eqref{problem}. 
Given $(c,h)\in\A$, we introduce, for each $n\in\N$, the truncated policies $(c^{(n)},h^{(n)})\in\A$:
\begin{equation}\label{truncated}
c^{(n)}_t := \left(\sum_{k=0}^{n-1} c_k(t) 1_{\{\tau_{k}\le t<\tau_{{k+1}}\}}\right) + c_n(t) 1_{\{t\ge \tau_{{n}}\}},\quad h^{(n)}_t := \left(\sum_{k=0}^{n-1} h_k(t) 1_{\{\tau_{k}\le t<\tau_{{k+1}}\}}\right) + h_n(t) 1_{\{t\ge \tau_{{n}}\}}.
\end{equation}

\begin{theorem}\label{thm:verification}
Let $w\in C^{1,1}(\R_+\times\R_+)$ satisfy \eqref{HJB}. Suppose for any $(x,m)\in\R^2_+$ and $(c,h)\in \A$,
\begin{align}
&\lim_{t\to\infty} \E\bigg[\exp\left(-\int_{\tau_n}^t (\delta+ M^{0,m,h^{(n)}}_s) ds\right)\cdot\nonumber\\
&\hspace{1.5in} w\left(\zeta^n X^{0,x,c^{(n)},h^{(n)}}_t,M^{0,m,h^{(n)}}_t\right)\ \bigg|\ Z_1,...,Z_n \bigg] = 0\quad \forall n\ge 0,\label{t to infty}\\
&\lim_{n\to\infty} \E\left[e^{-\delta\tau_{n}} w\left(\zeta^n X^{0,x,c^{(n)},h^{(n)}}_{\tau_n},M^{0,m,h^{(n)}}_{\tau_n}\right)\right] = 0.\label{n to infty}
\end{align}
\begin{itemize}
\item [(i)] $w(x,m) \ge V(x,m)$ on $\R_+\times\R_+$. 
\item [(ii)] Suppose there exist two measurable functions $\bar c$, $\bar h:\R^2_+\to\R_+$ such that
$\bar c(x,m)$ and $\bar h(x,m)$ are maximizers of 
\[
\sup_{c\ge 0}\left\{U(cx)-cxw_x(x,m)\right\}\quad \hbox{and}\quad \sup_{h\ge 0}\left\{ -w_m(x,m)g(h)-h x w_x(x,m)\right\},
\]
respectively, for all $(x,m)\in\R^2_+$. Let $\bar X$, $\bar M$, $\bar N$ denote the solutions to 
\begin{align*}
dX_s &= X_s[r - (\bar c(X_s,M_s)+\bar h(X_s,M_s))] ds\quad  X_0 =x,\\
dM_s &= M_s\left[\beta - g(\bar h(\zeta^{N_s} X_s,M_s))\right]ds\quad M_0=m,\\
N_s &= \sum_{k=0}^\infty k 1_{\{T_k\le t <T_{k+1}\}},\quad \hbox{with}\  T_0:=0,\ T_{n+1}:= \inf\left\{t\ge T_n\ \middle|\ \int_{T_n}^t M_s ds\ge Z_{n+1}\right\},\ n\ge 1.
\end{align*}
Define the processes $(\hat c,\hat h)$ by 
\begin{equation}\label{hat c h}
\hat c_t := \bar c(\zeta^{\bar N_t} \bar X_t, \bar M_t)\quad \hbox{and}\quad \hat h_t := \hat h(\zeta^{\bar N_t} \bar X_t,\bar M_t),\quad \hbox{for}\ t\ge 0.
\end{equation}
If $(\hat c,\hat h)\in\A$, then $(\hat c,\hat h)$ is an optimal control of the problem \eqref{problem}, and $w(x,m)=V(x,m)$ on $\R_+\times\R_+$. 
\end{itemize}
\end{theorem}

\begin{proof}
(i) Given $(c,h)\in\A$, recall that 
\[
c_t = \sum_{n=0}^\infty c_{n}(t) 1_{\{\tau_{n}\le t<\tau_{{n+1}}\}}\quad \hbox{and}\quad h_t = \sum_{n=0}^\infty h_{n}(t) 1_{\{\tau_{n}\le t<\tau_{n+1}\}}
\]
for some $\{c_{n}\}$, $\{h_{n}\}\in\mathfrak L$. We claim that the following holds for all $n\in \N$:
\begin{equation}\label{claim}
w(x,m)\ge \E\bigg[\int_{0}^{\tau_{{n}}} e^{-\delta t} U\left(c_t \zeta^{N_t}{X}^{0,x,c,h}_t\right)dt\bigg] +\E\left[e^{-\delta\tau_{{n}}}w\left(\zeta^{n} X^{0,x,c,h}_{\tau_{{n}}},M^{0,m,h}_{\tau_{{n}}}\right)\right].
\end{equation}
First, we prove this for $n=1$. 
Since $c_0$ and $h_0$ are deterministic functions and $w$ is a classical solution to \eqref{HJB}, 
\begin{align}\label{w>V}
&e^{-\int_0^t (\delta+M^{0,m,h_0}_\nu) d\nu} w(X^{0,x,c_0,h_0}_t,M^{0,m,h_0}_t)\le w(x,m)\nonumber\\
&\hspace{0.2in} -\int_0^t e^{-\int_0^s(\delta+M^{0,m,h_0}_\nu)d\nu}\left[U\left(c_0(s) X^{0,x,c_0,h_0}_s\right) + M^{0,m,h_0}_s w\left(\zeta X^{0,x,c_0,h_0}_s,M^{0,m,h_0}_s\right)\right]ds,
\end{align}
for all $t\ge 0$. Letting $t\to\infty$ and in view of \eqref{t to infty},
\begin{equation}
\begin{split}\label{1}
w(x,m)\ge & \int_0^\infty e^{-\int_0^s(\delta+M^{0,m,h_0}_\nu)d\nu} U\left(c_0(s) X^{0,x,c_0,h_0}_s\right) ds\\
&+ \int_0^\infty e^{-\int_0^s(\delta+M^{0,m,h_0}_\nu)d\nu} M^{0,m,h_0}_s w(\zeta X^{0,x,c_0,h_0}_s,M^{0,m,h_0}_s) ds.
\end{split}
\end{equation}
Thanks to Fubini's theorem and \eqref{tau>t}, observe that 
\begin{align}
\E\bigg[\int_0^{\tau_{1}} e^{-\delta t}U\left(c_t \zeta^{N_t}{X}^{0,x,c,h}_t\right) dt\bigg] &= \E\left[\int_0^\infty 1_{\{\tau_{1}>t\}} e^{-\delta t}U\left(c_0(t) X^{0,x,c_0,h_0}_t\right) dt\right]\nonumber\\
&=\int_0^\infty e^{-\int_0^t(\delta+M^{0,m,h_0}_\nu)d\nu} U\left(c_0(t) X^{0,x,c_0,h_0}_t\right) dt,\label{2}\\
\E\bigg[e^{-\delta \tau_{1}} w\left(\zeta^{} X^{0,x,c,h}_{\tau_{1}},M^{0,m,h}_{\tau_{1}}\right)\bigg]&= \E\left[e^{-\delta \tau_{1}} w\left(\zeta X^{0,x,c_0,h_0}_{\tau_{1}},M^{0,m_0,h_0}_{\tau_{1}}\right)\right]\nonumber\\
&\hspace{-0.5in}=\int_0^\infty e^{-\int_0^t(\delta+M^{0,m,h_0}_\nu)d\nu} M^{0,m,h_0}_t w\left(\zeta X^{0,x,c_0,h_0}_t,M^{0,m,h_0}_t\right) dt\label{3},
\end{align}
whence \eqref{claim} holds true for $n=1$ in view of \eqref{1}-\eqref{3}. Now, suppose \eqref{claim} holds true for $n=k>1$. That is,
\begin{equation}\label{claim'}
w(x,m)\ge \E\bigg[\int_{0}^{\tau_{{k}}} e^{-\delta t} U\left(c_t \zeta^{N_t}{X}^{0,x,c,h}_t\right)dt\bigg] +\E\left[e^{-\delta\tau_{{k}}}w\left(\zeta^{k} X^{0,x,c,h}_{\tau_{{k}}},M^{0,m,h}_{\tau_{{k}}}\right)\right].
\end{equation}
By writing $x_k = \zeta^{k}X^{0,x,c,h}_{\tau_{k}}$ and $m_k = M^{0,m,h}_{\tau_{k}}$,  
we get \eqref{w>V} with $(0,x,m,c_0,h_0)$ replaced by $(\tau_{k},x_k,m_k,c_k,h_k)$. This, together with \eqref{t to infty}, gives 
\begin{align}\label{180}
&w\left(\zeta^{k}X^{0,x,c,h}_{\tau_{{k}}},M^{0,m,h}_{\tau_{{k}}}\right) = w(x_k,m_k)\notag\\
&\hspace{0in}\ge  \E\bigg[\int_{\tau_{k}}^\infty \exp\left(-\int_{\tau_{k}}^t(\delta+M^{\tau_{k},m_k,h_k}_\nu)d\nu\right) U\left(c_k(t) X^{\tau_{k},x_k,c_k,h_k}_t\right) dt\ \bigg|\ Z_1,...,Z_k\bigg]\\
&\hspace{0in}+ \E\left[\int_{\tau_{k}}^\infty \exp\left(-\int_{\tau_{k}}^t(\delta+M^{\tau_{k},m_k,h_k}_\nu)d\nu\right) M^{\tau_{k},m_k,h_k}_t w\left(\zeta X^{\tau_{k},x_k,c_k,h_k}_t,M^{\tau_{k},m_k,h_k}_t\right) dt\ \middle|\ Z_1,...,Z_k\right].\notag
\end{align}
Using Fubini's theorem and \eqref{tau>t} as in \eqref{2}-\eqref{3}, the above inequality implies that
\begin{align}\label{181}
\E\left[e^{-\delta\tau_{k}}w\left(\zeta^{k}X^{0,x,c,h}_{\tau_{{k}}},M^{0,m,h}_{\tau_{{k}}}\right)\right]\ge\  &\E\bigg[\int_{\tau_{k}}^{\tau_{{k+1}}} e^{-\delta t} U\left(c_t \zeta^{N_t}{X}^{0,x,c,h}_t\right) dt\bigg]\notag\\
&+\E\left[ e^{-\delta \tau_{{k+1}}} w\left(\zeta^{k+1} X^{0,x,c,h}_{\tau_{{k+1}}},M^{0,m,h}_{\tau_{{k+1}}}\right)\right].
\end{align}
This, together with \eqref{claim'}, shows that
\[
w(x,m)\ge \E\bigg[\int_{0}^{\tau_{{k+1}}} e^{-\delta t} U\left(c_t \zeta^{N_t}{X}^{0,x,c,h}_t\right) dt\bigg] +\E\bigg[e^{-\delta\tau_{{k+1}}}w\bigg(\zeta^{k+1} X^{0,x,c,h}_{\tau_{{k+1}}},M^{0,m,h}_{\tau_{{k+1}}}\bigg)\bigg].
\]
The claim \eqref{claim} therefore holds by induction. Letting $n\to\infty$ in \eqref{claim}, by the monotone convergence theorem and \eqref{n to infty}, $w(x,m)\ge \E\big[\int_{0}^{\infty} e^{-\delta t} U\big(c_t \zeta^{N_t}{X}^{0,x,c,h}_t\big) dt\big]$ for all $(c,h)\in\A$. Taking the supremum over $(c,h)\in\A(m)$ leads to $w(x,m)\ge V(x,m)$.

(ii) With $(c,h)=(\hat c,\hat h)$, the inequality \eqref{w>V} turns into an equality, whence \eqref{claim} holds with equality. Sending $n\to\infty$, the monotone convergence theorem and \eqref{n to infty} imply that 
$$w(x,m)= \E\bigg[\int_{0}^{\infty} e^{-\delta t} U\left(\hat c_t \zeta^{N_t} {X}^{0,x,\hat c,\hat h}_t\right)dt\bigg] \le V(x,m).$$
This, together with part (i), shows that $w(x,m)=V(x,m)$ and $(\hat c,\hat h)$ is an optimal control.
\end{proof}

\thmref{thm:verification} can be extended to include the risky asset $S$ in \eqref{stock}. Recall the setup in Section~\ref{sec:risky}, especially $\A'$ in \eqref{A'} and the value function $V$ in \eqref{problem_0_risky asset}. For any $(c,h,\pi)\in \A'$, we can also consider the truncated version $(c^{(n)}, h^{(n)}, \pi^{(n)})\in \A'$ defined as in \eqref{truncated}.

\begin{theorem}\label{thm:verification with stock}
Let $w\in C^{1,1}(\R_+\times\R_+)$ satisfy \eqref{HJB_risky asset}. For any $(x,m)\in\R^2_+$ and $(c,h,\pi)\in \A'$, \eqref{t to infty} and \eqref{n to infty} hold, with $X^{0,x,c^{(n)},h^{(n)}}$ replaced by $X^{0,x,c^{(n)},h^{(n)},\pi^{(n)}}$ and $\E[\ \cdot \mid Z_1,...,Z_n]$ by $\E[\ \cdot\mid \F_{\tau_n}]$.
\begin{itemize}
\item [(i)] $w(x,m) \ge V(x,m)$ on $\R_+\times\R_+$. 
\item [(ii)] Suppose there exist measurable functions $\bar c$, $\bar h$, $\bar \pi:\R^2_+\to\R_+$, with $\bar c$ and $\bar h$ as described in Theorem~\ref{thm:verification} (ii) and $\bar\pi(x,m)$ being the maximizer of 
\[
\sup_{\pi\in \R}\left\{\pi\mu x w_x(x,m)+ \frac12 \sigma^2\pi^2 x^2 w_{xx}(x,m)\right\},\quad \forall\ (x,m)\in\R^2_+. 
\]
Let $\bar X$, $\bar M$, $\bar N$ denote the solutions to 
\begin{align*}
dX_s &= X_s[r + \mu \bar\pi(X_s,M_s)- (\bar c(X_s,M_s)+\bar h(X_s,M_s))] ds + \sigma X_s\bar\pi(X_s,M_s) dW_s,\quad  X_0 =x;\\
dM_s &= M_s \left[\beta - g(\bar h(\zeta^{N_s} X_s,M_s))\right]ds,\quad M_0=m;\\
N_s &= \sum_{k=0}^\infty k 1_{\{T_k\le t <T_{k+1}\}},\quad \hbox{with}\ T_0:=0,\ T_{n+1}:= \inf\left\{t\ge T_n\ \middle|\ \int_{T_n}^t M_s ds\ge Z_{n+1}\right\},\ n\ge 1.
\end{align*}
Consider the processes $(\hat c,\hat h, \hat \pi)$, with $\hat c$  and $\hat h$ defined as in \eqref{hat c h} and $\hat \pi_t:=\bar \pi(\zeta^{\bar N_t} \bar X_t, \bar M_t)$ for $t\ge 0$. If $(\hat c,\hat h,\hat\pi)\in\A'$, defined in \eqref{A'}, then $(\hat c,\hat h, \hat \pi)$ is an optimal control of the problem \eqref{problem_0_risky asset}, and $w(x,m)=V(x,m)$ on $\R_+\times\R_+$. 
\end{itemize}
\end{theorem}

\begin{proof}
(i) We follow the arguments in \thmref{thm:verification}, with $X^{0,x,c,h}$ replaced by $X^{0,x,c,h,\pi}$ in \eqref{X}. For any $(\hat c,\hat h,\hat\pi)\in\A'$, we now prove \eqref{claim} for all $n\in\N$. For $n=1$, as $w$ is a solution to \eqref{HJB_risky asset}, It\^{o}'s formula  yields \eqref{w>V}, with the left hand side and the second line under the expectation $\E_2$. By \eqref{t to infty}, letting $t\to\infty$ gives \eqref{1}, with the right hand side under the expectation $\E_2$. Fubini's theorem and \eqref{tau>t'} imply that \eqref{2}-\eqref{3}, with their second lines again under the expectation $\E_2$. Thus, \eqref{claim} holds for $n=1$. Now, suppose \eqref{claim} holds for $n=k>1$, i.e. \eqref{claim'} is true. As $w$ is a solution to \eqref{HJB_risky asset} and in view of \eqref{t to infty}, by It\^{o}'s formula \eqref{180} holds, with $\E[\ \cdot \mid Z_1,...,Z_n]$ replaced by $\E[\ \cdot\mid \F_{\tau_n}]$. By Fubini's theorem and \eqref{tau>t'} as above, \eqref{181} follows. This, together with \eqref{claim'}, implies that \eqref{claim} holds for $n=k+1$. Thus, \eqref{claim} follows by induction. Letting $n\to\infty$ in \eqref{claim} and recalling \eqref{n to infty}, the same argument as at the end of the proof of \thmref{thm:verification} (i) yields that $w(x,m)\ge V(x,m)$. 

(ii) This follows from the same argument as in the proof of \thmref{thm:verification} (ii).
\end{proof}

In the sequel we relax the conditions in \thmref{thm:verification}. To this end, for any $(x,m)\in\R^2_+$ and $(c,h)\in\A$, suppose that the household is given additional wealth $\eps >0$ at time $0$, and decides not to spend it at all over time. Imagine that at time $0$ the household deposits $x\ge 0$ in a standard account, and $\eps>0$ in a separate additional account. Then, the household behaves as if there was no additional wealth: at each time $t\ge 0$, the amount it spends in consumption (resp. healthcare) is $c_t$ (resp. $h_t$) multiplied by the standard account balance. The rates of spending in consumption and healthcare therefore become
\begin{equation}\label{eps-policies}
(c_\eps)_t = \frac{c_t X^{0,x,c,h}_t}{X^{0,x,c,h}_t+\eps e^{rt}},\quad (h_\eps)_t = \frac{h_t X^{0,x,c,h}_t}{X^{0,x,c,h}_t+\eps e^{rt}}\quad \forall t\ge 0.
\end{equation}
This new process $h_\eps$ of spending rate in healthcare, different from $h$, changes the moments of deaths.  More precisely, starting from time 0, the household takes 
\begin{equation}\label{eps-version}
h_0^\eps(t)  :=  \frac{h_{0}(t) X^{0,x,c_0,h_0}_t}{X^{0,x,c_0,h_0}_t+\eps e^{rt}}
\end{equation}
as instantaneous spending rates in healthcare. As in \eqref{tau's}, the time of the first death is defined as
\[
\tau^{\eps}_1 := \inf\left\{t\ge 0\ \middle|\ \int_{0}^t M^{0,m,h^\eps_0}_s ds \ge Z_1\right\}\le \tau_{1}.
\]
Starting from time $\tau^{\eps}_1$, the household takes 
\[
h_1^\eps(t)  :=  \frac{h_{0}(t) X^{0,x,c_0,h_0}_t}{X^{0,x,c_0,h_0}_t +\eps e^{rt}}1_{\{t<\tau_1\}} + \frac{h_{1}(t) X^{0,x,c^{(1)},h^{(1)}}_t}{X^{0,x,c^{(1)},h^{(1)}}_t +\eps e^{rt}}1_{\{t\ge \tau_1\}}
\]
as instantaneous spending rates in healthcare. 
Set $m^\eps_1:= M^{0,m,h^\eps_0}_{\tau^{\eps}_1}$, the time of the second death is defined as in \eqref{tau's} by
\[
\tau^{\eps}_2 := \inf\left\{t\ge \tau^{\eps}_1\ \middle|\ \int_{\tau^{\eps}_1}^t M^{\tau^{\eps}_1,m^\eps_1,h^\eps_1}_s ds \ge Z_2\right\}\le \tau_2.
\]
In general, for each $n\in \N$, the household, starting from time $\tau^{\eps}_{n}$, takes 
\[
h_n^\eps(t)  :=  \sum_{k=0}^{n-1} \frac{h_{k}(t) X^{0,x,c^{(k)},h^{(k)}}_t}{X^{0,x,c^{(k)},h^{(k)}}_t+\eps e^{rt}} 1_{\{\tau_{{k}}\le t<\tau_{{k+1}}\}} + \frac{h_{n}(t) X^{0,x,c^{(n)},h^{(n)}}_t}{X^{0,x,c^{(n)},h^{(n)}}_t+\eps e^{rt}} 1_{\{t\ge \tau_{{n}}\}}
\]
as instantaneous spending rates in healthcare. 
Set $m^\eps_n:= M^{\tau^{\eps}_{n-1},m^\eps_{n-1},h^\eps_{n-1}}_{\tau^{\eps}_n}$, the $(n+1)^{th}$ death moment is defined as in \eqref{tau's} by
\[
\tau^{\eps}_{n+1} := \inf\left\{t\ge \tau^{\eps}_n\ \middle|\ \int_{\tau^{\eps}_n}^t M^{\tau^{\eps}_n,m^\eps_n,h^\eps_n}_s ds \ge Z_{n+1}\right\}\le \tau_{{n+1}}.
\]
As in \eqref{counting process}, we can introduce the counting process 
\begin{equation}\label{counting process eps}
N^\eps_t := n\quad \hbox{for}\ t\in[\tau^\eps_{{n}},\tau^\eps_{{n+1}}).
\end{equation}
Similarly, define for each $n\in\N$,
\[
c_n^\eps(t)  :=  \sum_{k=0}^{n-1} \frac{c_{k}(t) X^{0,x,c^{(k)},h^{(k)}}_t}{X^{0,x,c^{(k)},h^{(k)}}_t+\eps e^{rt}} 1_{\{\tau_{{k}}\le t<\tau_{{k+1}}\}} + \frac{c_{n}(t) X^{0,x,c^{(n)},h^{(n)}}_t}{X^{0,x,c^{(n)},h^{(n)}}_t+\eps e^{rt}} 1_{\{t\ge \tau_{{n}}\}}.
\]
Observe that $\{c_n^\eps\}$, $\{h_n^\eps\}\in\mathfrak L$, and it can be checked that 
\[
(c_\eps)(t)  =  \sum_{k=0}^{\infty} c_n^\eps(t) 1_{\{\tau^\eps_{{k}}\le t<\tau^\eps_{{k+1}}\}},\quad (h_\eps)(t)  =  \sum_{k=0}^{\infty} h_n^\eps(t) 1_{\{\tau^\eps_{{k}}\le t<\tau^\eps_{{k+1}}\}}.
\]
This in particular shows that $(c_\eps,h_\eps)\in\A$. In view of \eqref{eps-policies}, we have the identity
\begin{equation}\label{identity}
{X}^{0,x+\eps,c_\eps,h_\eps}_t = {X}^{0,x,c,h}_t + \eps e^{rt}.
\end{equation}

\begin{proposition}\label{prop:verification}
Let $w\in C^{1,1}(\R_+\times\R_+)$ satisfy \eqref{HJB}. Suppose for any $(x,m)\in\R^2_+$, $(c,h)\in\A$, and $\eps>0$,
\begin{align}
&\lim_{t\to\infty} \E\left[\exp\left(-\int_{\tau^{\eps}_n}^t (\delta+ M^{0,m,h_\eps^{(n)}}_s) ds\right) w\left(\zeta^n X^{0,x+\eps,c_\eps^{(n)},h_\eps^{(n)}}_t,M^{0,m,h_\eps^{(n)}}_t\right)\ \middle|\ Z_1,...,Z_n\right]= 0\quad \forall n\ge 0,\label{t to infty'}\\
&\lim_{n\to\infty} \E\left[e^{-\delta\tau^{\eps}_n} w\left(\zeta^n X^{0,x+\eps,c_\eps^{},h_\eps^{}}_{\tau^{\eps}_n},M^{0,m,h_\eps^{}}_{\tau^{\eps}_n}\right)\right] = 0,\label{n to infty'}\\
&\lim_{\eps\to 0} \E\left[\int_{0}^{\infty} e^{-\delta t} U\left(c_t \zeta^{N^{\eps}_t} {X}^{0,x,c,h}_t\right)dt\right] = \E\left[\int_{0}^{\infty} e^{-\delta t} U\left(c_t \zeta^{N_t} {X}^{0,x,c,h}_t\right)dt\right]\label{eps to 0}.
\end{align}
\begin{itemize}
\item [(i)] $w(x,m)\ge V(x,m)$ on $\R_+\times\R_+$. 
\item [(ii)] Suppose the measurable functions $\bar h$ and $\bar c$ specified in \thmref{thm:verification} (ii) exist, so that we can define $(\hat c,\hat h)$ as in \eqref{hat c h}. If $(\hat c,\hat h)\in\A$ and satisfies \eqref{t to infty} and \eqref{n to infty}, then $(\hat c,\hat h)$ is an optimal control to the problem \eqref{problem}, and $w(x,m)= V(x,m)$ on $\R_+\times\R_+$.
\end{itemize}
\end{proposition}

\begin{proof}
We carry out the same arguments as in \thmref{thm:verification}. With the aid of \eqref{t to infty'}, we obtain \eqref{claim}, with $(x,c,h,\tau_{n},N_t)$ replaced by $(x+\eps,c_\eps,h_\eps,\tau^{\eps}_n,N^{\eps}_t)$. Letting $n\to \infty$ and using \eqref{n to infty'},
\[
w(x+\eps,m)\ge  \E\left[\int_{0}^{\infty} e^{-\delta t} U\left((c_\eps)_t \zeta^{N^{\eps}_t}{X}^{0,x+\eps,c_\eps,h_\eps}_t\right)dt\right]=\E\left[\int_{0}^{\infty} e^{-\delta t} U\left(c_t \zeta^{N^{\eps}_t} {X}^{0,x,c,h}_t\right)dt\right],
\]
where the equality follows from \eqref{identity} and the definition of $c_\eps$ in \eqref{eps-policies}. Sending $\eps\to 0$, we obtain from \eqref{eps to 0} that $w(x,m)\ge \E\big[\int_{0}^{\infty} e^{-\delta t} U\big(c_t \zeta^{N^{}_t} {X}^{0,x,c,h}_t\big)dt\big]$. Taking the supremum over $(c,h)\in\A$ leads to $w(x,m)\ge V(x,m)$. The proof of (ii) is the same as \thmref{thm:verification} (ii).
\end{proof}

Theorem~\ref{thm:verification with stock} can also be relaxed in a similar fashion. 

\begin{proposition}\label{prop:verification with stock}
Let $w\in C^{1,1}(\R_+\times\R_+)$ satisfy \eqref{HJB_risky asset}. For any $(x,m)\in\R^2_+$, $(c,h,\pi)\in \A'$, and $\eps>0$, suppose \eqref{t to infty'}, \eqref{n to infty'}, and \eqref{eps to 0} hold, with $X^{0,x,c^{(n)}_\eps,h^{(n)}_\eps}$replaced by $X^{0,x,c^{(n)}_\eps,h^{(n)}_\eps,\pi^{(n)}_\eps}$ and $\E[\ \cdot \mid Z_1,...,Z_n]$ by $\E[\ \cdot\mid \F_{\tau_n}]$. 
\begin{itemize}
\item [(i)] $w(x,m)\ge V(x,m)$ on $\R_+\times\R_+$. 
\item [(ii)] Suppose the measurable functions $(\bar c, \bar h, \bar\pi)$ specified in \thmref{thm:verification with stock} (ii) exist, so that we can define $(\hat c,\hat h, \hat \pi)$ therein. If $(\hat c,\hat h,\hat\pi)\in\A'$ and satisfies \eqref{t to infty} and \eqref{n to infty} as specified in \thmref{thm:verification with stock}, then $(\hat c,\hat h,\hat\pi)$ is an optimal control to the problem \eqref{problem_0_risky asset}, and $w(x,m)= V(x,m)$ on $\R_+\times\R_+$.
\end{itemize}
\end{proposition}

\end{document}